\documentclass[review]{elsarticle}
\usepackage[a4paper, margin=1.5cm]{geometry}

\usepackage{t1enc}
\usepackage{graphicx}
\usepackage{natbib}
\usepackage{longtable}
\usepackage[table,usenames,dvipsnames]{xcolor}
\usepackage{caption}
\usepackage{multirow}
\usepackage{hyperref}
\usepackage{rotating}
\usepackage{array}
\usepackage{hhline}

\journal{Journal of Mathematical Economics}

\biboptions{authoryear}

\newcolumntype{?}[1]{!{\vrule width #1}}

\newcommand{\vabs}{\widehat{v}}
\newcommand{\vrel}{v}

\newcommand{\uabs}{\widehat{u}}
\newcommand{\urel}{u}

\newcommand{\en}{\cellcolor{Gray!25}}
\definecolor{myBlue}{RGB}{120,200,210}
\definecolor{myGreen}{RGB}{0,180,0}
\definecolor{myRed}{RGB}{220,0,0}
\definecolor{myOrchid}{RGB}{190,125,219}

\newcommand\sign {\text{sign}}

\usepackage{amsmath}
\usepackage{amsthm}
\usepackage{amssymb}

\theoremstyle{plain}
\newtheorem{theorem}{Theorem}[section]
\newtheorem*{theorem*}{Theorem}
\newtheorem{lemma}{Lemma}[section]
\newtheorem{corollary}{Corollary}[section]

\theoremstyle{definition}
\newtheorem{example}{Example}[section]
\newtheorem{definition}{Definition}[section]
\newtheorem{remark}{Remark}[section]
\newtheorem*{remark*}{Remark}
\newtheorem*{assumption*}{Assumption}

\def\c#1{\multicolumn{1}{|c|}{#1}}  

\begin{document}

\begin{frontmatter}
\title{Resource-monotonicity and Population-monotonicity in Connected Cake-cutting}

\author[biu]{Erel Segal-Halevi}

\author[gt,cor]{Bal\'azs Sziklai}

\address[biu]{Bar-Ilan University, Ramat-Gan 5290002, Israel. erelsgl@gmail.com}

\address[gt]{'Momentum' Game Theory Research Group, Centre for Economic and Regional Studies, Hungarian Academy of Sciences. H-1112 Budapest Buda\"{o}rsi \'ut 45., Email: sziklai.balazs@krtk.mta.hu}

\address[cor]{Corvinus University of Budapest, Department of Operations Research and Actuarial Sciences, H-1093 Budapest F\H{o}v\'am t\'er 8.}

\begin{abstract}
In the classic cake-cutting problem (Steinhaus, 1948), a heterogeneous resource has to be divided among $n$ agents with different valuations in a \emph{proportional} way ---
giving each agent a piece with a value of at least $1/n$ of the total. In many applications, such as dividing a land-estate or a time-interval, it is also important that the pieces are \emph{connected}.
We propose two additional requirements: resource-monotonicity (RM) and population-monotonicity (PM). When either the cake or the set of agents changes and the cake is re-divided using the same rule, the utility of all remaining agents must change in the same direction. Classic cake-cutting protocols are neither RM nor PM.
Moreover, we prove that no Pareto-optimal proportional division rule can be either RM or PM.
Motivated by this negative result, we search for division rules that are \emph{weakly-Pareto-optimal} --- no other division is strictly better for all agents.
We present two such rules. The \emph{relative-equitable} rule, which assigns the maximum possible relative value equal for all agents, is proportional and PM. The so-called \emph{rightmost mark} rule, which is an improved version of the Cut and Choose protocol, is proportional and RM for two agents.
\end{abstract}


\end{frontmatter}
\section{Introduction}
Monotonicity axioms have been extensively studied with respect to cooperative game theory \citep{Calleja2012}, political representation \citep{Balinski1982}, computer resource allocation \citep{Ghodsi2011Dominant} and many other fair division problems
\cite[chapters 3 6 7]{Moulin2004Fair}, \cite[chapter 7]{Thomson_2011}.

These axioms express the idea of solidarity among agents: whenever the environment changes in a way that requires the re-allocation of resources, the welfare of all agents not responsible for the change should be affected in the same direction --- either they should all be made at least as well off as they were initially, or they should all be made at most as well off. This is the so called replacement principle which was formulated by \cite{Thomson1997Replacement}.

Two common monotonicity axioms are \emph{resource monotonicity} (RM) and \emph{population monotonicity} (PM).
Resource-monotonicity, sometimes known as aggregate monotonicity, requires that when new resources are added, and the same division rule is used consistently, the utility of all agents should weakly increase.
Population-monotonicity is concerned with changes in the number of participants. It requires that when someone leaves the division process and abandons his share, the utility of the remaining participants should weakly increase. Conversely, when a new agent joins the process, all existing participants should participate in supporting the new agent, thus their utility should weakly decrease.

Monotonicity axioms are sometimes conceived as more important than other, more basic fairness axioms.
A prominent example is the practical problem of \emph{apportionment}: there is a parliament with a fixed number of seats and administrative regions with different number of voters. The seats have to be distributed among the regions in such way that the resulting allotment ensures proportional representation. The solution originally employed in the USA congress was the Hamiltonian rule, which guaranteed proportional representation. However, it was found that this rule exhibits the \emph{Alabama paradox} --- increasing the number of seats would have rendered state Alabama with less seats. In other words, the rule violates resource-monotonicity. Later, it was found that this rule also exhibits the \emph{new state paradox} (the \emph{``Oklahoma paradox''}) --- the addition of Oklahoma to the USA would have rendered state Maine with more seats. In other words, the rule also violates population-monotonicity. These violations pressed legislators to adopt a new apportionment method. The currently used method, the so called Huntington-Hill method, fails to uphold Hare-quota, a basic guarantee of proportionality, but it is satisfies the monotonicity axioms \citep{Balinski1975,caulfield2008apportioning}.

The present paper studies these two monotonicity requirements in the framework of the \emph{fair cake-cutting} problem \citep{Steinhaus1948}, where a single heterogeneous resource - such as land or time - has to be divided fairly. We approach the problem from the classic point of view, when each agent is interested in getting a connected piece. Connected utility functions make sense in many applications. E.g., when dividing land, a large connected piece can be used for building a house, but a collection of small disconnected patches of land is virtually useless. Similarly when departments dispute over the availability of a conference room, each of them is interested in reserving the room for a contiguous period which is free of disruption. Another example is a long TV ad, which needs to be aired in one piece. These examples show that assuming connected utilities in some cases is a reasonable restriction, and indeed, many cake-cutting papers explicitly assume that each agent must be allocated a connected piece.

\subsection{Results}
We survey many traditional cake-cutting protocols and show that they do not satisfy either of the monotonicity axioms. In particular, all methods based on the Cut and Choose scheme violate both resource-monotonicity and population-monotonicity.
This motivates a search for division rules that are both fair in the conventional sense and monotonic. We conducted this search under two different assumptions regarding the agents' utility functions, which are equally common in the cake-cutting literature.

In both models, each agent has a value measure defined over the cake. In the \emph{additive} model, the utility of a piece of cake is just the value measure of that piece; the geometry of the piece has no importance. In the \emph{connected} model, the utility of a piece of cake is the value of the most valuable connected component of the piece.

Our results for the additive model can be found in another manuscript \citep{ourArxivPaperAdditive}. These results were mainly positive: we found several Pareto-optimal proportional division rules that satisfy one or both monotonicity axioms. In particular, the \emph{Nash-optimal} rule, maximizing the product of values, is envy-free (hence also proportional), resource-monotonic and population-monotonic.

The present paper studies the connected model. Here, the situation is not so positive. Each of the monotonicity properties is incompatible with proportionality and Pareto-optimality. That is, no Pareto-optimal proportional division rule can be either resource- or population-monotonic. Thus the fair divider has to choose between Pareto-optimality and monotonicity. While from an economics perspective Pareto-optimality is crucial, public opinion may not always agree. In some cases people are willing to sacrifice efficiency to get fairness \citep{Herreiner2009}. As a compromise, we suggest several division rules which are proportional and \emph{weakly}-Pareto-optimal (no other allocation is strictly preferred by all agents; see e.g.\ \cite{Varian1974Equity}) while satisfying one of the monotonicity axioms. The \emph{max-equitable-connected} rules, which give equally-valuable pieces to each agent while maximizing this value, are both population monotonic.
There are two such rules: the rule equalizing the relative values (normalized such that the entire cake value is 1) is proportional but not resource-monotonic, and the rule equalizing the absolute (not normalized) values is resource-monotonic but not proportional. Additionally, we present a proportional and resource-monotonic division protocol for two agents. It is an open question whether there exists a weakly-Pareto-optimal, proportional and resource-monotonic rule for $n$ agents.

The equitable rule belongs to the cardinal welfarism framework (cf.\ chapter 3 of \cite{Moulin2004Fair}). It relies on inter-agent utility comparison, and makes sense if and when such a comparison is feasible. For example, suppose the agents are firms, each of which wants to use the land to a pre-specified purpose (e.g.\ one firm plans to dig for oil, another firm wants to build housing complexes, etc.). Then, economic models can be used to estimate the monetary utility of each firm for each piece of land and the estimates can be used to calculate equitable divisions (see the conclusion of \cite{Chambers_2005} for further discussion of the additive utility model).

The paper is organized as follows. Section \ref{sec:literature} reviews the related literature. Section \ref{sec:model} formally presents the cake-cutting problem and the monotonicity axioms. Section \ref{sec:classic-protocols} examines classic cake-cutting protocols and shows that they are not monotonic. Sections \ref{sec:negative} and \ref{sec:positive} present our negative and positive results. Section \ref{sec:conclusion} concludes and presents a table summarizing the various rules' properties.

\section{Related Work}\label{sec:literature}
The cake-cutting problem originates from the work of the Polish mathematician Hugo Steinhaus and his students Banach and Knaster \citep{Steinhaus1948}. Their primary concern was how to divide the cake in a fair way. Since then, game theorists analyzed the strategic issues related to cake-cutting, while computer scientists were focusing mainly on how to implement solutions, i.e.\ the computational complexity of cake-cutting protocols.

Many economists regard land division as an important application of division procedures \citep[e.g. ][]{Berliant1988Foundation,Berliant1992Fair,Legut1994Economies,Chambers2005Allocation,DallAglio2009Disputed,Husseinov2011Theory,Nicolo2012Equal}).
Hence, they note the importance of imposing some geometric constraints on the pieces allotted to the agents. Connectivity is the most well-studied constraint.

As we already noted in the introduction there is a vast literature on monotonicity related issues. To our knowledge our paper is the first that explicitly defines RM and PM for the cake cutting setting. However, there are a few other axioms which bear resemblance to these two.

\cite{Arzi2011} study the "dumping paradox" in cake-cutting. They show that, in some cakes, discarding a part of the cake improves the total social welfare of any envy-free division. This implies that enlarging the cake might decrease the total social welfare. This is related to resource-monotonicity; the difference is that in our case we are interested in the welfare of the individual agents and not in the total social welfare.

\cite{Chambers_2005} studies a related cake-cutting axiom called "division independence": if the cake is divided into sub-plots and each sub-plot is divided according to a rule, then the outcome should be identical to dividing the original land using the same rule. He proves that the only rule which satisfies Pareto-optimality and division independence is the utilitarian-optimal rule - the rule which maximizes the sum of the agents' utilities. The rule is only feasible when the utilities are additive (with no connectivity constraints). Unfortunately, this rule does not satisfy fairness axioms such as proportionality.

\cite{Walsh2011Online} studies the problem of "online cake-cutting", in which agents arrive and depart during the process of dividing the cake. He shows how to adapt classic procedures like cut-and-choose and the Dubins-Spanier in order to satisfy online variants of the fairness axioms. Monotonicity properties are not studied, although the problem is similar in spirit to the concept of population-monotonicity.

Finally, we mention that the consistency axiom (cf.\ \cite{Young1987} or \cite{Thomson_2012}) is related to population-monotonicity, but it is fundamentally different as in that case the leaving agents take their fair shares with them.

\subsection{Equitable divisions}
The ``heroes'' of the present paper are the equitable division rules. The equitability condition in cake-cutting is much less studied than other properties such as proportionality and envy-freeness. Some notable exceptions are presented below.

The first proof to the \emph{existence} of an equitable division is implied by the seminal work of \citet{Dubins_1961}. The piece allocated to each agent can be an arbitrary member of a $\sigma$-algebra, i.e, not necessarily connected. The result of \citet{Alon_1987} implies the existence of an equitable division with a limited number of cuts, but still not necessarily connected. Max-relative-equitable divisions without the connectivity requirement were studied extensively by \cite{DallAglio2001DubinsSpanier} (he calls such divisions \emph{equi-optimal}).
Max-relative-equitable divisions with the connectivity requirement for two agents were studied by \cite{Jones2002}. The generalization for $n$ agents was mentioned by \citet{BramsJonesKlamler2006,BramsJonesKlamler2007}. They related the problem of equitable-connected cake-cutting to a set of integral equations, but did not prove they are solvable.
The latter point was discussed by  \citet{mawet2010equitable} for the special case when the valuations are piecewise-constant and everywhere-positive. \citet{Aumann2010Efficiency} proved that equitable-connected allocations exist for general valuations. \citet{Cechlarova2013Existence} extended this result and proved that an equitable-connected division exists for any ordering of the agents, and for at least one ordering it is also proportional.

The \emph{computability} of equitable allocations is discussed by several recent works. \citet{Cechlarova2012Computability} proved that there is no finite discrete procedure for finding an allocation that is equitable, connected and proportional. \citet{Procaccia2017} showed that this impossibility holds even without the connectivity and proportionality requirements. On the positive side, \cite{Cechlarova2011Near,Cechlarova2012Computability} provided discrete procedures that attain \emph{$\epsilon$-equitable connected divisions} --- divisions in which the difference between the value of every two agents is at most $\epsilon$.
Independently and contemporaneously to our work, \citet{Branzei2017Query} presented a moving-knife procedure for equitable cake-cutting, for the special case in which all players are ``hungry'' (i.e, all valuations are strictly positive).

The main contribution of the present paper to the literature on equitable division is in showing its advantages over other, more famous cake-cutting procedures. In particular, we show that it is population monotonic, and can be made either proportional or resource-monotonic depending on whether relative or absolute values are used.

A secondary contribution is a moving-knives procedure for finding an equitable-connected division in any ordering of the agents, which is applicable for general valuations (not only strictly positive). This does not contradict the impossibility results mentioned above, since a moving-knife procedure is continuous rather than discrete.

\section{Model}  \label{sec:model}
\subsection{Cake-cutting}
A cake-cutting problem is a triple $\Gamma(N,C,(\vabs_i)_{i \in N})$ where:

\begin{itemize}
  \item $N=\{1,2,\dots,n\}$ denotes the set of agents who participate in the cake-cutting process. In examples with a small number of agents, we often refer to them by names (Alice, Bob, Carl...).
  \item $C$ is the cake. For simplicity we assume that $C$ is a interval, $C=[0,c]$ for some real number $c$. We call a Borel subset of $C$ a \emph{slice}.
  \item $\vabs_i$ is the value measure of agent $i$. It is a finite real-valued function defined on the Borel subsets of $[0,\infty)$.
\end{itemize}


As the term ``measure'' implies, the value measures of all agents are countably additive: the value measures of a union of disjoint slices is the sum of the values of the slices. Moreover, we assume that the value measures are non-negative and bounded. That is, $\vabs_i$ assigns a non-negative, but finite number to each slice of $C$.
We also assume that the value measures are absolutely-continuous with respect to Lebesgue measure: this means that a slice with zero length has zero value to anyone. Therefore it is unimportant to specify which agent gets the endpoints of an interval, since the endpoints have zero value.
All these assumptions are standard in the cake-cutting literature.

Our model diverges from the standard cake-cutting setup in that we do not require the value measures to be normalized. That is, the value of the entire cake is not necessarily the same for all agents. This is important because we examine scenarios where the cake changes, so the cake value might become larger or smaller. Hence, we differentiate between \emph{absolute} and \emph{relative} value measures:
\begin{itemize}
\item The \textbf{absolute} value measure of the entire cake, $\vabs_i(C)$, can be any positive value and it can be different for different agents.
\item The \textbf{relative} value of the entire cake is 1 for all agents.
Relative value measures are denoted by $\vrel_i$ and defined by:
$\vrel_i(S):={\vabs_i(S) / \vabs_i(C)}$.
\end{itemize}

It is also common to assume that value measures are private information of the agents. This question leads us to whether agents are honest about their preferences. Cake-cutting problems can be studied from a strategic angle, however, the results are mostly negative. For example, in any deterministic discrete strategy-proof protocol, there always exists an agent that gets the empty piece \citep{branzei2015dictatorship}.
Here, we will not analyze the strategic behavior of the agents but assume they act truthfully.

The \emph{utility} of an agent is based on its value measure. In the present paper we assume that:
\begin{align*}
\uabs_i(X)=\sup_{I\subset X} \vabs_i(I)
&&
\urel_i(X)=\sup_{I\subset X} \urel_i(I)
\end{align*}
where the supremum is over all connected intervals $I$ that are subsets of $X$. That is, an agent can only use a single connected piece.


The aim is to divide the cake into $n$ pairwise-disjoint slices. A \emph{division rule} is a correspondence that takes a cake-cutting problem as input and returns a \emph{division} $X=(X_1,\dots, X_n)$, or a set of divisions. Note that a division does not necessarily compose a partition of $C$ (i.e.\ free disposal is assumed).

Since all agents have connected utilities, we can assume without loss of generality that each agent receives a connected piece, i.e, for all $i$, $X_i$ is an interval. Under this assumption, $\uabs_i(X_i)=\vabs_i(X_i)$ and $\urel_i(X_i)=\vrel_i(X_i)$ for all $i$, so from now on we will use only $\vabs_i$ and $\vrel_i$.

A division rule $R$ is called \emph{essentially single-valued (ESV)} if $X,Y \in R(\Gamma)$ implies that for all $i \in N$, $\vabs_i(X_i)=\vabs_i(Y_i)$. That is, even if $R$ returns a set of divisions, all agents are indifferent between these divisions.

The classic requirements of fair cake-cutting are the following. A division $X$ is called:

\begin{itemize}
  \item \emph{Pareto-optimal} (PO) if there is no other division which is weakly better for all agents and strictly better for at least one agent.
  \item \emph{Weakly-Pareto-optimal} (WPO) if there is no other division which is strictly better for all agents.
  \item \emph{Proportional} (PROP) if each agent gets at least $1/n$ fraction of the cake according to his own evaluation, i.e.\ for all $i \in N$, $\vrel_i(X_i)\geq 1/n$. Note that the  definition uses relative values.
  \item \emph{Envy-free} (EF) if each agent gets a piece which is weakly better, for that agent, than all other pieces: for all $i,j \in N$, $\vrel_i(X_i)\geq \vrel_i(X_j)$. Note that here it is irrelevant whether absolute or relative values are used. Note also that PO+EF imply PROP.
\end{itemize}

A division rule is called \emph{Pareto-optimal} (PO) if it returns only PO divisions. The same applies to WPO, PROP and EF.

\subsection{Monotonicity}
We now define the two monotonicity properties. In the introduction we defined them informally for the special case in which the division rule returns a single division. Our formal definition is more general and applicable to rules that may return a set of divisions.

\begin{definition}
Let $N$ be a fixed set of agents, $C=[0,c]$, $C'=[0,c']$ two cakes where $c<c'$, and $(\vabs_i)_{i \in N}$ value measures on $[0,\infty)$. The cake-cutting problem $\Gamma'=(N, C', (\vabs_i)_{i \in N})$ is called a \emph{cake-enlargement} of the problem $\Gamma=(N, C, (\vabs_i)_{i \in N})$.
\end{definition}

By definition the cake is always enlarged on the right hand side. This might be critical for some protocols. For instance in the Dubins-Spanier moving knife protocol the cake is processed from left to right \citep{Dubins_1961}.   However, most of our results (except that of Subsection \ref{sub:right-mark}) are valid whenever $C\subset C'$, regardless of whether the cake is enlarged from the left, right or middle.

\begin{definition}
(a) A division rule $R$ is called \emph{upwards resource-monotonic}, if for all pairs $(\Gamma, \Gamma')$, where $\Gamma'$ is a cake-enlargement of $\Gamma$, for every division $X\in R(\Gamma)$ there \emph{exists} a division $Y \in R(\Gamma')$ such that $\vabs_i(Y_i) \geq \vabs_i(X_i)$ for all $i \in N$ (i.e\ all agents are weakly better-off in the new division).

(b) A division rule $R$ is called \emph{downwards resource-monotonic}, if for all pairs $(\Gamma', \Gamma)$, where $\Gamma'$ is a cake-enlargement of $\Gamma$, for every division $Y\in R(\Gamma')$ there \emph{exists} a division $X \in R(\Gamma)$ such that $\vabs_i(X_i) \leq \vabs_i(Y_i)$ for all $i \in N$ (i.e\ all agents are weakly worse-off in the new division).

(c) A division rule is \emph{resource-monotonic} (RM), if it is both upwards and downwards resource-monotonic.
\end{definition}

\begin{definition}
Let $C$ be a fixed cake, $N$ and $N'$ two sets of agents such that $N\supset N'$ and $(\vabs_i)_{i \in N}$ their value measures.
The cake-cutting problem $\Gamma'=(N', C, (\vabs_i)_{i \in N'})$ is called a \emph{population-reduction} of the problem $\Gamma=(N, C, (\vabs_i)_{i \in N})$.
\end{definition}

\begin{definition}
(a) A division rule $R$ is called \emph{upwards population-monotonic}, if for all pairs $(\Gamma', \Gamma)$ such that $\Gamma'$ is a population-reduction of $\Gamma$, for every division $Y \in R(\Gamma')$ there exists a division $X \in R(\Gamma)$ such that $\vabs_i(X_i) \leq \vabs_i(Y_i)$ for all $i \in N'$ (all the original agents are weakly worse-off in the new division).

(b) A division rule $R$ is called \emph{downwards population-monotonic}, if for all pairs $(\Gamma, \Gamma')$ such that $\Gamma'$ is a population-reduction of $\Gamma$, for every division $X \in R(\Gamma)$ there exists a division $Y \in R(\Gamma')$ such that $\vabs_i(Y_i) \geq \vabs_i(X_i)$ for all $i \in N'$ (all remaining agents are weakly better-off in the new division).

(c) A division rule is \emph{population-monotonic} (PM), if it is both upwards and downwards population-monotonic.
\end{definition}

\begin{remark}
As usual in the literature, the monotonicity axioms care only about absolute values. In other words, it is not considered a violation of RM if the relative value of an agent decreases when the cake grows.
\end{remark}

\begin{remark} For essentially-single-valued solutions, downwards resource (or population) monotonicity implies upwards resource (or population) monotonicity and vice versa. Set valued solutions, however, may satisfy only one direction of these axioms.
\end{remark}

\begin{remark}
The monotonicity axioms in \cite{Thomson_2011} require that \emph{all} divisions in $R(\Gamma)$ have to be weakly better/worse than all divisions in $R(\Gamma')$. In contrast, our definition only requires that there \emph{exists} such a division. This is closer to the definition of aggregate monotonicity, which originates from cooperative game theory \citep{Peleg2007}. The rationale is that even if a set-valued solution is used, only a single allocation will be implemented. Hence, the divider can be faithful to the monotonicity principles even if the rule suggests many non-monotonic allocations as well.

Because our monotonicity requirements are weaker, any impossibility result in our model is valid in Thomson's model, too. This is not true in general for positive results; however, the specific positive results in the present paper are all based on essentially-single-valued rules. Hence, by the previous remark, they are valid in Thomson's model, too.
\end{remark}

\section{Monotonicity of classic cake-cutting protocols} \label{sec:classic-protocols}
Although resource- and population-monotonicity are well established axioms in various fields of fair division, the cake-cutting literature has not adopted these ideas so far.
Moreover, classical division methods like the Banach-Knaster \citep{Steinhaus1948}, Cut and Choose, Dubins-Spanier \citep{Dubins_1961}, Even-Paz \citep{Even1984}, Fink \citep{Fink1964} or the Selfridge-Conway protocol do not satisfy these axioms. A detailed explanation for most of these can be found in \citep{Brams1996}. For completeness, our survey includes procedures that return disconnected pieces.

All the counterexamples below feature \emph{piecewise homogeneous} cakes. These are finite unions of disjoint intervals, such that on each interval the value densities of all agents are constant (although different agents may evaluate the same piece differently). In such cases, the function $\vabs_i([0,x])$ --  which displays the value (for agent $i$) of the piece which lies left to the point $x\in \mathbb{R}$ -- is a piecewise-linear function (see Figure \ref{plvf4}). Piecewise-homogeneous cakes are interesting on their own (see e.g.\ \cite{Aziz2014Cake}), however, we stress that our results, hold for arbitrary cakes - not only for piecewise-homogeneous ones.

\begin{figure}
  \centering
  \includegraphics[width=13.5cm]{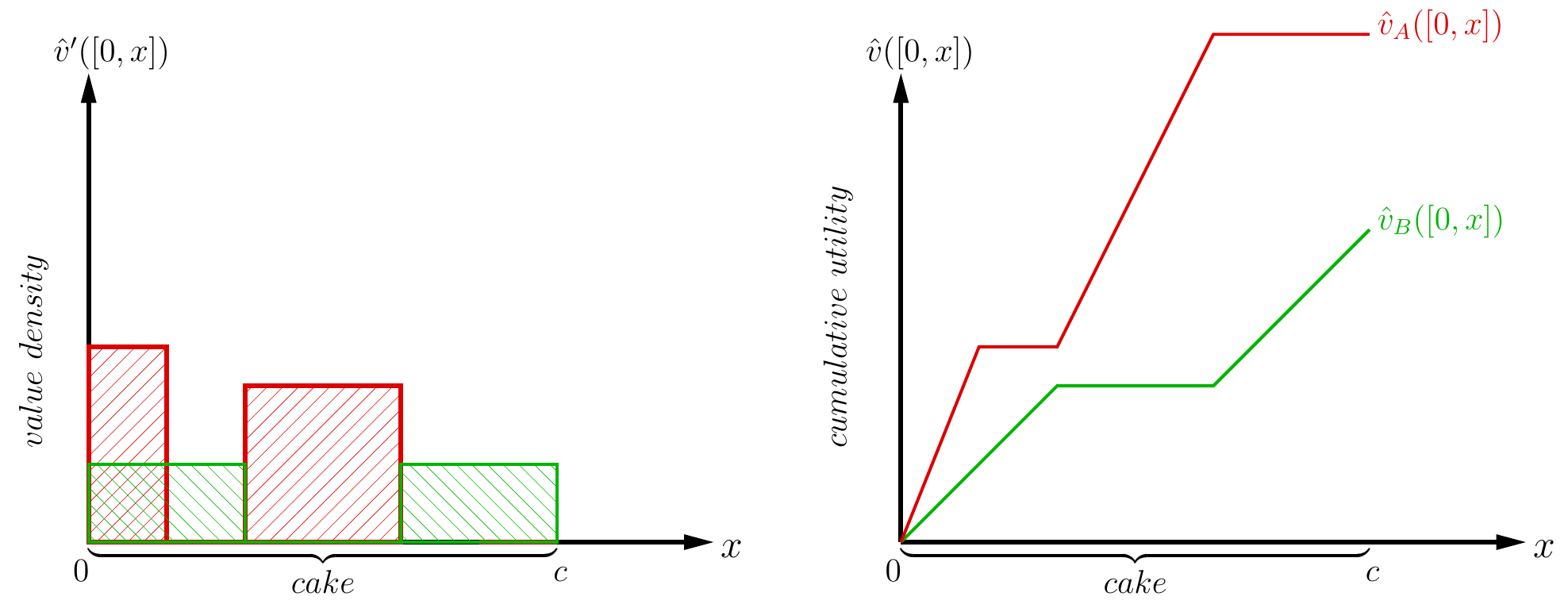}\\
  \caption{Piecewise homogeneous cake with two players. The derivative of $\hat{v}([0,x])$ indicates the density of the value. Note that the value measures are not normalized, hence $\vabs_A([0,c])\neq \vabs_B([0,c])$.}\label{plvf4}
\end{figure}

Piecewise homogeneous cakes can be represented by a simple table containing the value densities of the agents on the different slices. For example the cake in Figure \ref{plvf4} has the following representation.

\begin{table}[h!]
  \centering
  \begin{tabular}{|l||c|c|c|c|c|c|}
    \hline
    $\vabs_A$ & 2.5 & 0 & 2  & 2 & 0  & 0 \\
    \hline
    $\vabs_B$ & 1 & 1 & 0  & 0 & 1 & 1 \\
    \hline
  \end{tabular}
  \end{table}

\subsection{Resource-monotonicity}
First let us examine the Cut and Choose protocol for two agents.
We define the \emph{cut-and-choose rule} as the rule in which one pre-specified agent (say, Alice) cuts the cake into two pieces equal in her eyes, the other agent (Bob) picks the piece that he prefers, and the first agent receives the remaining piece. The following example shows that this rule is not resource-monotonic. In the examples below, the $\blacktriangledown$ sign over a column indicates the enlargement, and the colored cells in an agent's row indicate the agent's piece.

\begin{center}
\begin{tabular}{lcccc}
     &  &    & & \\
\hline
\multicolumn{1}{|c||}{$\vabs_A$} &
\multicolumn{1}{|c}{\cellcolor{myGreen!25}1} &
\multicolumn{1}{|c}{\cellcolor{myGreen!25}1} &
\multicolumn{1}{|c}1  & \multicolumn{1}{|c|}1
\\
\hline
\multicolumn{1}{|c||}{$\vabs_B$} &
\multicolumn{1}{|c}1 &
\multicolumn{1}{|c}1 & \multicolumn{1}{|c}{\cellcolor{myBlue!25}$3$}  & \multicolumn{1}{|c|}{\cellcolor{myBlue!25}$3$}
\\
\hline
\end{tabular}
  \quad\quad
 \begin{tabular}{lccccc}
       &  &    &  & &  $\blacktriangledown$ \\
\hline
\multicolumn{1}{|c||}{$\vabs_A$} &
\multicolumn{1}{|c}{\cellcolor{myGreen!25}1} &
\multicolumn{1}{|c}{\cellcolor{myGreen!25}1} &
\multicolumn{1}{|c}{\cellcolor{myGreen!25}1} &
\multicolumn{1}{|c}1 &
\multicolumn{1}{|c|}2
\\
\hline
\multicolumn{1}{|c||}{$\vabs_B$} &
\multicolumn{1}{|c}1 &
\multicolumn{1}{|c}1 &
\multicolumn{1}{|c}{$3$}&
\multicolumn{1}{|c}{\cellcolor{myBlue!25}$3$} &
\multicolumn{1}{|c|}{\cellcolor{myBlue!25}2}
\\
\hline
\end{tabular}
\end{center}
When the extra piece is not present (left), Alice cuts the cake after the second slice, allowing Bob to choose the piece worth $6$ for him. However, when the cake is enlarged, Alice cuts after the third slice and Bob's utility drops to $5$.

This example implies that the Banach-Knaster, Dubins-Spanier, Even-Paz and the Fink methods are not resource-monotonic either, as they all produce the same divisions on the above cake as Cut and Choose.%
\footnote{A more recent protocol, the Recursive Cut and Choose, proposed by \cite{Tasnadi2003}, violates resource-monotonicity for the same reason.}

Finally we examine the Selfridge-Conway envy-free protocol for three agents. This protocol has a pre-specified \emph{cutter} (who cuts the cake to three equal pieces) and a pre-specified \emph{trimmer} (who trims his best piece to make it equal to his second-best piece). W.l.o.g, we define the \emph{Selfridge-Conway-rule} as the rule in which Alice is the cutter and Bob is the trimmer. The following example shows that this rule is not RM.

\begin{center}
  \centering
    \begin{tabular}{lcccccc}
    & &  &     \\
    \hline
    \multicolumn{1}{|c||}{$\vabs_A$} &
    \multicolumn{1}{|c}4 &
    \multicolumn{1}{|c?{1mm}}2 &
    \multicolumn{1}{c}2 &
    \multicolumn{1}{|c?{1mm}}4 &
    \multicolumn{1}{c}{\cellcolor{myGreen!25}4} & \multicolumn{1}{|c|}{\cellcolor{myGreen!25}2}  \\
    \hline
    \multicolumn{1}{|c||}{$\vabs_B$} &\multicolumn{1}{|c}{\cellcolor{myBlue!25}5} & \multicolumn{1}{|c?{1mm}}{\cellcolor{myBlue!25}2} & \multicolumn{1}{c}3&
    \multicolumn{1}{|c?{1mm}}4&
    \multicolumn{1}{c}1&
    \multicolumn{1}{|c|}1   \\
    \hline
    \multicolumn{1}{|c||}{$\vabs_C$} &
    \multicolumn{1}{|c}1 &
    \multicolumn{1}{|c?{1mm}}2&
    \multicolumn{1}{c}{\cellcolor{myOrchid!25}4} &
    \multicolumn{1}{|c?{1mm}}{\cellcolor{myOrchid!25}4}& \multicolumn{1}{c|}1&
    \multicolumn{1}{|c|}1   \\
    \hline
  \end{tabular}
  \quad\quad
      \begin{tabular}{lccccccc}
      & &  &  & & & &  $\blacktriangledown$ \\
    \hline
    \multicolumn{1}{|c||}{$\vabs_A$} &
    \multicolumn{1}{|c}4 &
    \multicolumn{1}{|c}2 &
    \multicolumn{1}{|c?{1mm}}2 &
    \multicolumn{1}{c}4 &
    \multicolumn{1}{|c?{1mm}}4 & \multicolumn{1}{c}{\cellcolor{myGreen!25}2} & \multicolumn{1}{|c|}{\cellcolor{myGreen!25}6}\\
    \hline
    \multicolumn{1}{|c||}{$\vabs_B$} & \multicolumn{1}{|c?{.5mm}}{\cellcolor{myBlue!25}5} & \multicolumn{1}{|c}{2} &
    \multicolumn{1}{|c?{1mm}}{3} &
    \multicolumn{1}{c}{4} &
    \multicolumn{1}{|c?{1mm}}1 &
    \multicolumn{1}{c}{1} &
    \multicolumn{1}{|c|}1 \\
    \hline
    \multicolumn{1}{|c||}{$\vabs_C$} &
    \multicolumn{1}{|c}1 &
    \multicolumn{1}{|c}2 &
    \multicolumn{1}{|c?{1mm}}{4}&
    \multicolumn{1}{c}{\cellcolor{myOrchid!25}4}& \multicolumn{1}{|c?{1mm}}{\cellcolor{myOrchid!25}1} &
    \multicolumn{1}{c}{1} &
    \multicolumn{1}{|c|}1 \\
    \hline
  \end{tabular}
\end{center}
In the smaller cake (left), Alice cuts the cake into three parts worth 6 to her, each made of two adjacent slices. The two most valuable parts for Bob are worth the same (7) so he passes. Then the agents choose in the order Carl, Bob, Alice. Carl's utility is 8.

In the larger cake (right), Alice cuts three parts worth 8 to her, made of 3, 2 and 2 slices. The leftmost part is most valuable for Bob, so he trims it to make it equal to the middle part. Carl takes the uncut part, which is worth 5 for him. Now, Carl divides the remainder to 3 equal pieces and the agents choose in order: Bob, Alice, Carl. Carl receives a piece worth at most 2, so his total utility is at most 7.

\subsection{Population-monotonicity}
Population-monotonicity is not applicable to protocols with fixed number of agents, such as Cut and Choose and Selfridge-Conway.
The following example shows that the Dubins-Spanier moving-knife protocol is not PM.
In the next couple of examples the cells of the leaving player are colored gray.

\begin{table}[h!]
  \centering
    \begin{tabular}{lcccccc}
    \hline
    \multicolumn{1}{|c||}{$\vabs_A$} &
    \multicolumn{1}{|c}{\cellcolor{myGreen!25}20} &
    \multicolumn{1}{|c}1 &
    \multicolumn{1}{|c}1 &
    \multicolumn{1}{|c}1 &
    \multicolumn{1}{|c}{10} &
    \multicolumn{1}{|c|}{27}
    \\
    \hline
     \multicolumn{1}{|c||}{$\vabs_B$} &  \multicolumn{1}{|c}1 &
     \multicolumn{1}{|c}{\cellcolor{myBlue!25}20} &
     \multicolumn{1}{|c}{\cellcolor{myBlue!25}10} &
     \multicolumn{1}{|c}{28} &
     \multicolumn{1}{|c}1 &
     \multicolumn{1}{|c|}1
     \\
    \hline
    \multicolumn{1}{|c||}{$\vabs_C$} &
    \multicolumn{1}{|c}1 &
    \multicolumn{1}{|c}1 &
    \multicolumn{1}{|c}{18} &   \multicolumn{1}{|c}{\cellcolor{myOrchid!25}{10}} &
    \multicolumn{1}{|c}{\cellcolor{myOrchid!25}{29}} &
    \multicolumn{1}{|c|}{\cellcolor{myOrchid!25}1}\\
    \hline
  \end{tabular}
    \quad\quad
      \begin{tabular}{lcccccc}
    \hline
    \multicolumn{1}{|c||}{$\vabs_A$} & \multicolumn{1}{|c}{20} &
    \multicolumn{1}{|c}1 &
    \multicolumn{1}{|c}1 &
    \multicolumn{1}{|c}1 & \multicolumn{1}{|c}{\cellcolor{myGreen!25}10} &
    \multicolumn{1}{|c|}{\cellcolor{myGreen!25}27} \\
    \hline
    \multicolumn{1}{|c||}{\en $\vabs_B$} & \multicolumn{1}{|c}{\en 1} & \multicolumn{1}{|c}{\en 20} & \multicolumn{1}{|c}{\en 11} & \multicolumn{1}{|c}{\en 28} & \multicolumn{1}{|c}{\en 1} & \multicolumn{1}{|c|}{\en 1}\\
    \hline
    \multicolumn{1}{|c||}{$\vabs_C$} &
    \multicolumn{1}{|c}{\cellcolor{myOrchid!25}1} &
    \multicolumn{1}{|c}{\cellcolor{myOrchid!25}1}&
    \multicolumn{1}{|c}{\cellcolor{myOrchid!25}18} &
    \multicolumn{1}{|c}{\cellcolor{myOrchid!25}10} &
    \multicolumn{1}{|c}{29} &
    \multicolumn{1}{|c|}{1}\\
    \hline
  \end{tabular}
\end{table}
When all three agents are present, Alice is the first to stop the knife and get a piece. In the second round, Bob stops the knife and gets a piece, and finally Carl receives the reminder which is worth for him 40.
However, if Bob is not present then Carl will be the first to stop the knife and his value will be 30.

The Even-Paz method and the Banach-Knaster protocol (when agents are ordered Alice-Bob-Carl) produce the same allocations. Thus, none of these three methods is population-monotonic.

Consider now the Fink procedure. This procedure was specifically designed with upwards-population-monotonicity in mind: when a new agent joins an existing division, he takes a proportional share from each of the existing agents, so all existing agents are weakly worse off; they all participate in supporting the new agent, which is what PM is all about. However, the Fink procedure is not downwards-PM, as the following example shows:

\begin{table}[h!]
  \centering
    \begin{tabular}{lccccccc}
    \hline
    \multicolumn{1}{|c||}{$\vabs_A$} & \multicolumn{1}{|c}{\cellcolor{myGreen!25}2} & \multicolumn{1}{|c}{\cellcolor{myGreen!25}2} & \multicolumn{1}{|c}{2} & \multicolumn{1}{|c}{2} &  \multicolumn{1}{|c}{1} &\multicolumn{1}{|c}{1} & \multicolumn{1}{|c|}{2} \\
    \hline
    \multicolumn{1}{|c||}{$\vabs_B$} & \multicolumn{1}{|c}0 & \multicolumn{1}{|c}0 & \multicolumn{1}{|c}0 & \multicolumn{1}{|c}{\cellcolor{myBlue!25}4} & \multicolumn{1}{|c}{\cellcolor{myBlue!25}2} & \multicolumn{1}{|c}{\cellcolor{myBlue!25}2} & \multicolumn{1}{|c|}4  \\
    \hline
    \multicolumn{1}{|c||}{$\vabs_C$} & \multicolumn{1}{|c}0 & \multicolumn{1}{|c}0 & \multicolumn{1}{|c}{\cellcolor{myOrchid!25}0} & \multicolumn{1}{|c}2 & \multicolumn{1}{|c}1 & \multicolumn{1}{|c}1 & \multicolumn{1}{|c|}{\cellcolor{myOrchid!25}2} \\
    \hline
  \end{tabular}
  \quad\quad
    \begin{tabular}{lccccccc}
    \hline
    \multicolumn{1}{|c||}{\en $\vabs_A$} & \multicolumn{1}{|c}{\en 2} & \multicolumn{1}{|c}{\en 2} & \multicolumn{1}{|c}{\en 2} & \multicolumn{1}{|c}{\en 2} &  \multicolumn{1}{|c}{\en 1} &\multicolumn{1}{|c}{\en 1} & \multicolumn{1}{|c|}{\en 2} \\
    \hline
    \multicolumn{1}{|c||}{$\vabs_B$} & \multicolumn{1}{|c}{\cellcolor{myBlue!25}0} & \multicolumn{1}{|c}{\cellcolor{myBlue!25}0} & \multicolumn{1}{|c}{\cellcolor{myBlue!25}0} & \multicolumn{1}{|c}{\cellcolor{myBlue!25}4} & \multicolumn{1}{|c}{\cellcolor{myBlue!25}2} & \multicolumn{1}{|c}{2} & \multicolumn{1}{|c|}4  \\
    \hline
    \multicolumn{1}{|c||}{$\vabs_C$} & \multicolumn{1}{|c}0 & \multicolumn{1}{|c}0 & \multicolumn{1}{|c}{0} & \multicolumn{1}{|c}2 & \multicolumn{1}{|c}1 & \multicolumn{1}{|c}{\cellcolor{myOrchid!25}1} & \multicolumn{1}{|c|}{\cellcolor{myOrchid!25}2} \\
    \hline
  \end{tabular}
\end{table}

Suppose that initially Alice and Bob use Cut and Choose and Bob is the chooser. He is able to salvage the whole cake according to his own evaluation. Now they divide their pieces into three equal parts, and Carl gets to choose one slice from each of them. Hence, Bob ends up with a piece worth at least 8 for him. But if Alice leaves, then Bob and Carl have to redivide the cake using Cut and Choose. Then, no matter who cuts, Bob ends up with only 6.

The above example seemingly contradicts our claim that upwards-PM implies downwards-PM and vice versa for single valued solutions. However, there is a subtle difference here. The Fink procedure is based on a predefined order of the agents, and it is only upwards monotonic if the new agent is the last in the order. An alternative explanation is to treat the Fink rule as a set-valued rule, which returns $n!$ possible allocations, for all $n!$ possible orderings of the agents.  Under this definition, the Fink rule is upwards-PM, but not downwards-PM as shown in the example.

\section{Negative Results}
\label{sec:negative}

When the agents do not care about connectivity,
the ideal division rule (at least in terms of fairness) is the \emph{Nash-optimal rule}, which maximizes the product of utilities: it is RM, PM, PO and EF, hence also proportional \citep{ourArxivPaperAdditive}.
Moreover, for every Nash-optimal allocation there exists a price-vector that is a competitive-equilibrium from equal-incomes (CEEI). Therefore, it makes sense to ask whether Nash-optimality and/or CEEI have all these desirable properties with connectivity too. Unfortunately, the answer is no.

Consider first the Nash-optimal rule and the following cake:
	\begin{center}
		\begin{tabular}{|l||c|c|c|c|c|c|}
			\hline  $\vabs_A$ & 2 & 2 & 2 & 2 & 2 & 2  \\
			\hline  $\vabs_B$ & 1 & 1 & 4 & 4 & 1 & 1  \\
			\hline
		\end{tabular}
	\end{center}
Without connectivity, it gives the two central slices to Bob and the four peripheral slices to Alice. The Nash-welfare is 8*8=64. The allocation is EF and PROP.
Moreover, it is supported in a competitive-equilibrium from equal-incomes, in which the price of a central slice is 2 and the price of a peripheral slice is 1 and the income of both agents is 4.

In contrast, with connectivity, the Nash-optimal rule is not proportional. To see this, observe that in both of the connected proportional divisions, Alice and Bob each get three slices and a value of 6, so the Nash-welfare is 36. However, when Bob gets four slices and Alice two slices, the Nash-welfare is 40. Hence neither of the two possible proportional allocations is Nash-optimal.

Moreover, with connectivity, a CEEI allocation might not exist at all. Recall that any CEEI allocation is both PO and EF; in following cake, no PO+EF allocation exists:
\begin{example}
\label{exm:no-peef}
[No connected allocation is both PO and EF]:\\
\begin{center}
\begin{tabular}{|l||c|c|c|c|c|c|c|}
\hline     $\vabs_A$ & 2 & 0 & 3 & 0 & 2 & 0 & 0 \\
\hline $\vabs_B$ & 0 & 0 & 0 & 0 & 0 & \cellcolor{blue!30}{7} & 0 \\
\hline     $\vabs_C$ & 0 & 2 & 0 & 2 & 0 & 0 & 3 \\
\hline
\end{tabular}
\end{center}
EF requires to give Bob a part of his 7 slice;
PO then requires to give him his entire 7 slice. Carl's piece can then be either at Bob's left or at Bob's right:
\begin{itemize}
\item If Carl is at Bob's left and his utility is 2, then the allocation is not PO since Carl can be moved to the rightmost slice and get a utility of 3 without harming any other agent.
\item If Carl is at Bob's left and his utility is more than 2, then Alice envies him since he holds her 3 slice while here utility is at most 2.
\item If Carl is at Bob's right, then by PO the entire left is given to Alice, but then Carl envies her.
\end{itemize}
\qed
\end{example}
Since both the Nash-optimal and the CEEI rules fail in the presence of connectivity requirements, we have to look for different rules.
But first we show that, with connectivity requirements, Pareto-optimality and proportionality are incompatible with resource-monotonicity:
\begin{theorem}
\label{exm:con-prop-poca-am-incompatible}
When there are two or more agents with connected utilities, any division rule which is proportional and Pareto-optimal cannot be resource-monotonic.
\end{theorem}
\begin{proof}
Consider the following cake, where the enlargement is marked by the $\blacktriangledown$ sign:

\begin{center}
	\begin{tabular}{lccccc}
         & &   &   \\
		\hline
		\c{$\vabs_A$} & \c{6} & \c{0} & \c{1} & \c{1}  \\
		\hline
		\c{$\vabs_B$} & \c{0} & \c{4} & \c{2} & \c{2}  \\
		\hline
	\end{tabular}
\quad\quad
	\begin{tabular}{lccccc}
         & &   & & & $\blacktriangledown$ \\
		\hline
		\c{$\vabs_A$} & \c{6} & \c{0} & \c{1} & \c{1} & \c{6} \\
		\hline
		\c{$\vabs_B$} & \c{0} & \c{4} & \c{2} & \c{2} & \c{0}  \\
		\hline
	\end{tabular}
\end{center}
In the smaller cake, any PROP+PO rule must give the leftmost slice to Alice and the rest of the cake to Bob. Hence, Alice's utility is 6 and Bob's utility is 8. In the larger cake, any PROP rule must give Alice a utility of at least 7. This leaves Bob a utility of at most 6.
\end{proof}

Moreover, Pareto-optimality and proportionality are incompatible with population-monotonicity, too:

\begin{theorem}
\label{exm:con-prop-poca-pm-incompatible}
	When agents have connected utilities, any division rule which is proportional and Pareto-optimal cannot be population-monotonic.
\end{theorem}
\begin{proof}
Consider again the cake of Example \ref{exm:no-peef}:
\begin{center}
\centering
\begin{tabular}{|l||c|c|c|c|c|c|c|}
\hline     $\vabs_A$ & 2 & 0 & 3 & 0 & 2 & 0 & 0 \\
\hline \en $\vabs_B$ & \en 0 & \en 0 & \en 0 & \en 0 & \en 0 & \en 7 & \en 0 \\
\hline     $\vabs_C$ & 0 & 2 & 0 & 2 & 0 & 0 & 3 \\
\hline
\end{tabular}
\end{center}
Note that all agents value the entire cake as 7, so by PROP each agent must receive a connected piece with a value of at least $2\frac{2}{3}$. Bob's piece must be in his 7 slice (by PROP) and most contain all this slice and nothing more (by PO).

Carl's piece can be either left or right of Bob's 7 slice. If Carl's piece is at Bob's left, then by PROP it must contain the two 2 slices of Carl. This leaves Alice a utility of at most 2, which violates PROP. So Carl's piece must be the Bob's right and Carl's utility is 3.

This leaves the entire region at Bob's left to Alice. By PO, she receives this entire region and her utility is 7.

Suppose Bob leaves. Now $n=2$, so Carl must get a value of at least $7/2 = 3.5$, so his piece must touch his middle "2" slice. But this leaves Alice a utility of at most 5.
\end{proof}

Despite the crucial importance of Pareto-optimality in economics, in our case it is problematic: it is incompatible with the stronger fairness criterion of envy-freeness. Even with the weaker criterion of proportionality, it is incompatible with any of the monotonicity axioms.

If we believe that a division rule is fair only if it satisfies both proportionality and monotonicity, we must compromise on efficiency. Some possible compromises are presented in the following section.

\section{Positive Results}\label{sec:positive}
\subsection{\textbf{Exactly-proportional rule: PROP+RM+PM}} \label{sub:exact-prop}
Our first division rule is both resource- and population-monotonic, but very inefficient. We present it merely as a benchmark for comparison with the more advanced rules that come later.

\begin{definition}
A division $X$ is called \emph{exactly-proportional} if it gives every agent a relative value of exactly $1/n$. Formally:
$\forall i\in N: \vrel_i(X_i)=1/n$.
\end{definition}

The \emph{exactly-proportional rule} returns an exactly-proportional division; such a division can be found, for example, using the following variant of the Banach-Knaster procedure \citep{Steinhaus1948}:
\begin{itemize}
\item Every agent marks a point $x_i$ such that $\vabs_i([0,x_i])=\frac{1}{n}$.
\item The procedure selects the leftmost point $x_{min}$ (breaking ties arbitrarily)\ and gives $[0,x_{min}]$ to the agent that made that mark.
\item The remaining agents divide the cake recursively in the same way (keeping the fraction $1/n$ fixed).
\item The cake that remains after the $n$-th step is discarded.
\end{itemize}

\begin{theorem}
The exactly-proportional rule is proportional and resource-monotonic and population-monotonic, but not weakly-Pareto-optimal.
\end{theorem}
\begin{proof}
PROP is obvious by definition.

RM holds because when the cake grows/shrinks, all agents receive the same fraction of a larger/smaller whole.

PM holds because when an agent leaves/joins, the remaining agents receive a larger/smaller fraction of the same whole.

The following cake shows that the rule is not WPO:
	\begin{center}
		\begin{tabular}{|l|c|c|}
			\hline
			$\vabs_A$ & 2 & 0 \\
			\hline
			$\vabs_B$ & 0 & 2\\
			\hline
		\end{tabular}
	\end{center}
An exactly-proportional division must give each agent a utility of exactly 1, yet it is possible to give each agent a utility of 2.
\end{proof}
In essence, the exactly-proportional rule tells the agents ``keep your happiness at the minimum proportional level of exactly $1/n$, so that when new resources become available, you can only become happier''. This guarantees PROP and RM and PM, but it is very inefficient.

\subsection{\textbf{Relative-equitable rule: WPO+PM+PROP}} \label{sub:equitable}
We present a population-monotonic division rule based on the notion of \emph{equitable cake divisions}. The idea of an equitable division is that all agents are equally happy --- each agent receives a piece with the same personal value.

\begin{definition}
(a) A cake division $X$ is called \emph{relative-equitable} if all agents receive exactly the same relative value. Formally: $\vrel_i(X_i) = \vrel_j(X_j)$ for all $i,j \in N$. This value is called the \emph{relative-equitable value} of the division.

(b) A relative-equitable division is called \emph{max-relative-equitable} if its relative-equitable value is weakly larger than of all relative-equitable divisions.

\end{definition}

\begin{definition}
(a) An \emph{agent-ordering}, denoted by $\pi$, is a permutation on the set of agents $N$.

(b) A connected partition of the cake into $n$ intervals is called a \emph{$\pi$-partition} if the intervals are assigned to the $n$ agents in the order specified by $\pi$.
\end{definition}

For example, a $\overline{132}$-partition is a connected partition in which Agent 1 receives the leftmost piece, Agent 3 receives the middle piece and Agent 2 receives the rightmost piece.

Independently of ours, some of the following lemmata were proved by \citet{Cechlarova2013Existence}. We mention each lemma that was previously proved. For completeness, we provide alternative proofs.

\begin{lemma}
\label{lemma:equitable-single-permutation}
For every agent-ordering $\pi$, there exists a relative-equitable $\pi$-partition.
\end{lemma}
\begin{proof}
A proof using generalized-inverse functions is given in Theorem 5 of \citet{Cechlarova2013Existence}.

A new and shorter proof, using the Borsuk-Ulam theorem, is given in Appendix \ref{sub:max-eq-proof}.
\end{proof}

Below, we present a moving-knife procedure that finds an equitable-connected division for any ordering of the agents. Note that an equitable-connected division cannot be found by a discrete procedure \citep{Cechlarova2012Computability}, so a moving-knife procedure is a natural alternative.

Without loss of generality, assume that $\pi$ is the order $1,\dots,n$. The procedure starts with the agents holding their knives at the leftmost end of the cake. There is a large screen where the current relative-equitable value is displayed, which is zero at the beginning. During the procedure the positions of the knives determine a division of the cake: the piece allotted to Agent $i$ is the piece between the knife of Agent $i$ and the knife of Agent $i-1$ (or for $i=1$ the leftmost end of the cake). The value on the screen increases continuously and all agents moves their knives to the right such that the value of each piece matches the value on the screen. This goes on until one of the following two things happen:

\begin{description}
  \item[(a)] The rightmost knife reaches the end of the cake.
  \item[(b)] The knife of an agent reaches the leftmost endpoint of an interval in which the value density of that particular agent is 0.
\end{description}

In case (a), the procedure stops and we have obtained an relative-equitable partition of the entire cake. The relative-equitable value is the value on the screen.

In case (b), the value on the screen is frozen temporarily and the procedure enters its second phase. Let $j$ be the rightmost agent whose piece is adjacent to a zero-value interval. So every agent who comes after $j$ in the predefined order can strictly increase the value of his piece by moving his knife to the right. We ask all agents starting from $j$ to move their knives to the right such that their value remains constant. This goes on until either (a) holds, or another agent $k>j$ reaches an interval of zero measure so (b) holds for that agent, or Agent $j$ reaches the end of his zero-measure interval. In the first case the procedure stops, in the second case the procedure continues at second phase with agent $k$, in the third case the procedure goes back to the first phase.

Since the rightmost knife is moving continuously and monotonically to the right, eventually it reaches the end of the cake and an equitable division is found.\footnote{Here we implicitly use the assumption that the value measures are bounded. If $\vabs_i([0,c])$ were infinite then the rightmost knife could move to the right indefinitely without reaching the end of the cake by slowing down.}
\qed

\begin{figure}[!h]
  \centering
  \includegraphics[width=13.5 cm]{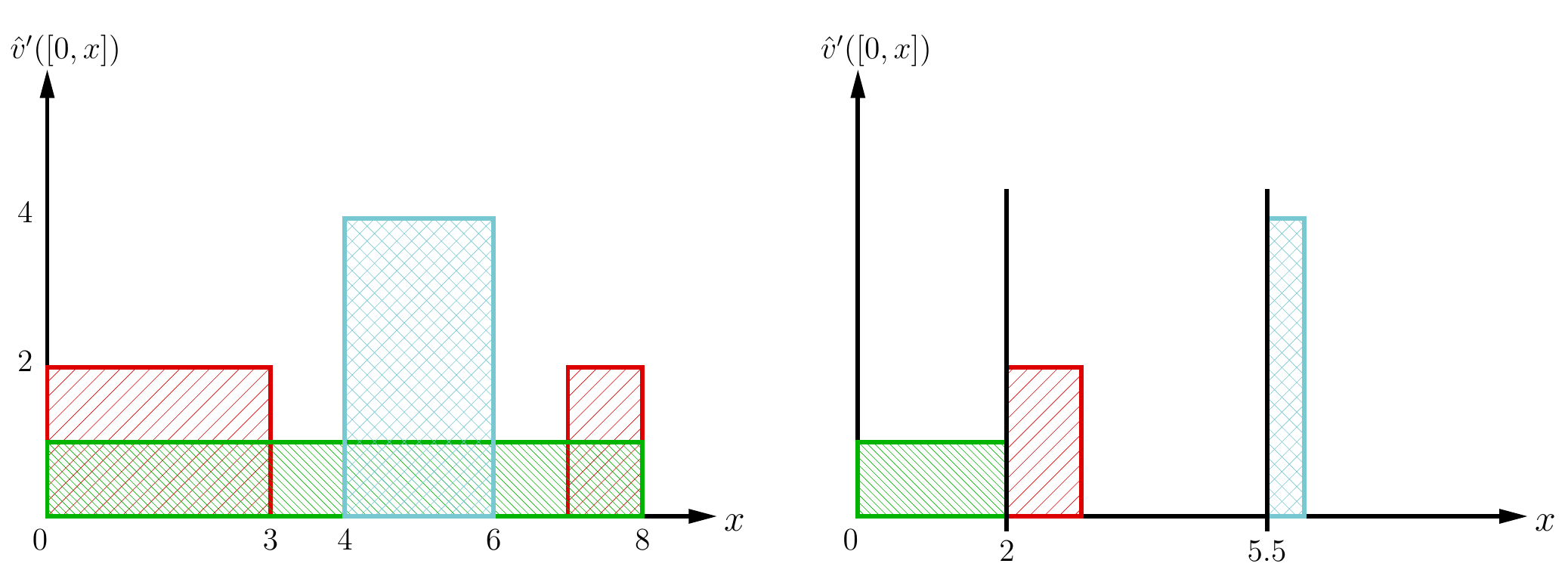}\\
  \caption{The moving knife procedure described in Lemma \ref{lemma:equitable-single-permutation} applied in a cake-division with three agents: Green (densely striped), Red (sparsely striped), Blue (grid).
  \textbf{Left:} the agents' valuations.
  \textbf{Right:} the resulting equitable division
  }

  \label{max-equitable-figure2}
\end{figure}

We demonstrate the somewhat informal description of the above moving-knife procedure on an example.
\begin{example}
Consider the piecewise homogeneous cake depicted in Figure \ref{max-equitable-figure2}. Three agents: Green, Red and Blue seek an equitable division of the cake, which has total value of 8 for each of them (this is a special case where the same relative value indicates the same absolute value). They agree on using the above procedure with the order Green, Red, Blue. Immediately at the beginning, we are at case (b) because Blue's knife is at a zero-value region. Thus, we enter phase 2 and Blue's knife moves to $x=4$. Then we return to phase 1.

As the knives move to the right, the agents increase the value of their pieces until they reach a relative value of $2/8$. At that moment Green's knife rests at $x=2$, Red's knife at $x=3$ and Blue's knife at $x=4.5$. The value displayed at the screen becomes fixed at this point since Red reached a zero-value interval. As Red gradually increases his piece, Blue moves his knife to the right making sure his value does not change. This continues until Blue himself reaches $x=6$, which is the start of an interval of zero value. Red stops his knife at $x=5.5$, but Blue continues until he reaches the right end of the cake. The resulting division is relative-equitable with value 2/8.
\end{example}

Let $X$ be a certain division of a cake. We denote the smallest relative value obtained by an agent by $\vrel^X_{min}$ and the largest by $\vrel^X_{max}$, such that for all $i=1,\dots,n$: $\vrel^X_{min}\leq \vrel_i(X_i)\leq \vrel^X_{max}$. Note that $\vrel^X_{min} = \vrel^X_{max}$ if and only if $X$ is a relative-equitable division.

\begin{lemma}
\label{lemma:equitable-min-max-value}
Let $\pi$ be an agent-ordering and $X$ a $\pi$-partition of a cake. Let $Y$ be a relative-equitable $\pi$-partition of the same cake, having a relative-equitable value $\vrel^Y$. Then: $\vrel^X_{min} \leq \vrel^Y \leq \vrel^X_{max}$.
\end{lemma}
\begin{proof}
The proof that $\vrel^X_{min} \leq \vrel^Y$
is in Lemma 7 of \cite{Cechlarova2013Existence};
the proof that $\vrel^Y \leq \vrel^X_{max}$ (which is analogous) is in Lemma 1 of \cite{Cechlarova2012Computability}.
We provide an alternative, graphic proof that that $\vrel^X_{min} \leq \vrel^Y$ (see Figure \ref{fig:equitable-min-max-value}).
\begin{figure}
  \centering
  \includegraphics[width=11cm]{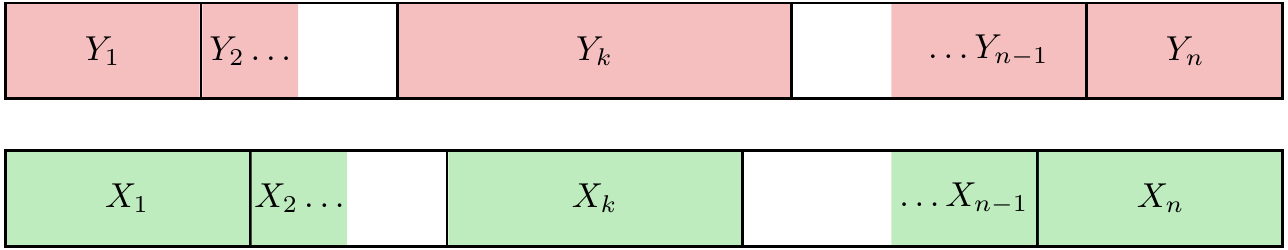}\\
  \caption{Proof of Lemma \ref{lemma:equitable-min-max-value}. \label{fig:equitable-min-max-value}}
\end{figure}

Assume w.l.o.g.\ that $\pi$ is the ordering $1,\dots,n$. Assume by contradiction that $\vrel^Y < \vrel^X_{min}$. In particular, this means that Agent 1 receives a smaller value in partition $Y$ than in partition $X$, that is, $\vrel_1(Y_1)<\vrel_1(X_1)$. Hence, the cut-point between pieces $Y_1$ and $Y_2$ is to the left of the cut-point between pieces $X_1$ and $X_2$.

The same is true for the $n$-th agent: $\vrel_n(Y_n)<\vrel_n(X_n)$. Hence the cut-point between pieces $Y_{n-1}$ and  $Y_n$ is to the \emph{right} of the cut-point between pieces $X_{n-1}$ and $X_n$. Because the leftmost cut-point moved to the left and the rightmost cut-point moved to the right, there must be a pair of adjacent cut-points such that the left one moved to the left and the right one moved to the right (see Figure \ref{fig:equitable-min-max-value}). Hence, there must be an index $k$, such that:
\begin{itemize}
\item The left boundary of piece $Y_k$ is to the left of the left boundary of $X_k$, and
\item The right boundary of  $Y_k$ is to the right of the right boundary of $X_k$.
\end{itemize}
This means that $Y_k \supset X_k$ which in turn implies that $\vrel_k(Y_k)\geq \vrel_k(X_k)$. This contradicts our assumption that $\vrel^Y < \vrel^X_{min}$.
\end{proof}

For every agent-ordering $\pi$, there may be many different equitable $\pi$-partitions. However, all these divisions have the same equitable value:

\begin{lemma}
\label{lemma:equitable-unique-value}
For every agent-ordering $\pi$, there is a unique value $\vrel^{\pi}$ which is the relative-equitable value in all relative-equitable $\pi$-partitions.
\end{lemma}
\begin{proof}
This simple corollary is also proved in Corollary 2 of \citet{Cechlarova2012Computability}.

Assume that there are two relative-equitable $\pi$-partitions: $X$ with equitable value $\vrel^X$ and $Y$ with equitable value $\vrel^Y$. By Lemma \ref{lemma:equitable-min-max-value}, $\vrel^X \leq \vrel^Y \leq \vrel^X$. Hence $\vrel^X=\vrel^Y$.
\end{proof}

A straightforward corollary of the above lemmata is that the orderings can be sorted by their equitable value. Since for $n$ agents there are finitely many different orderings the following holds.
\begin{corollary}
There exist max-relative-equitable divisions with connected pieces.
\end{corollary}
Define the \emph{relative-equitable rule} as the rule that returns all connected max-relative-equitable divisions of the cake. By Lemma \ref{lemma:equitable-unique-value}, this rule is essentially-single-valued.
\begin{lemma}
\label{lemma:equitable-is-wpo}
The relative-equitable division rule is weakly-Pareto-optimal.
\end{lemma}
\begin{proof}
Let $Y$ be a max-relative-equitable division with equitable value $\vrel^Y$. Suppose by contradiction that there is a division $X$ in which the utility of all agents is strictly more than $\vrel^Y$. Let $\pi$ be the agent ordering in $X$. By Lemma \ref{lemma:equitable-min-max-value}, the relative-equitable-value of the relative-equitable division in ordering $\pi$ is at least $\vrel^X_{min} > \vrel^Y$. But this contradicts the maximality of $Y$.
\end{proof}

\begin{lemma}
\label{lemma:equitable-is-pm}
The relative-equitable division rule is population-monotonic.
\end{lemma}
\begin{proof}
Since the rule is essentially-single-valued, it is sufficient to prove downwards-PM.

Let $X$ be a max-relative-equitable for $n$ agents with equitable value $\vrel^X$. Suppose that an agent $i\in N$ abandons his share. Give agent $i$'s piece to an agent that holds an adjacent piece, e.g.\ to agent $i+1$. Call the resulting division $Y$. We obtained a connected division for $n-1$ agents, in which the smallest value enjoyed by an agent is at least $\vrel^X$ (indeed, the value of all agents except $i+1$ is exactly $\vrel^X$, and the value of agent $i+1$ is at least as large). By Lemma \ref{lemma:equitable-min-max-value}, the maximum equitable value in the new situation is at least $\vrel^X$. Hence, in the max-relative-equitable for $n-1$ agents, the value of all agents is at least as large as in the previous division.
\end{proof}

\begin{lemma}
\label{lem:relative-equitable}
The relative-equitable division rule is proportional.
\end{lemma}
\begin{proof}
This is also proved in Corollary 1 of \citet{Cechlarova2012Computability}.

Let $X$ be any connected proportional division of the cake (by \cite{Steinhaus1948} such a division always exists). Because $X$ is proportional, $\vrel^X_{min}\geq 1/n$. Hence, by Lemma \ref{lemma:equitable-min-max-value}, the value of a relative-equitable division in the same ordering as $X$ is at least $1/n$. Hence, the maximum relative-equitable value is at least $1/n$.
\end{proof}

Unfortunately, the relative-equitable rule is not resource-monotonic.

\begin{example} \label{exm:relative-equitable-not-rm}
Consider the following cake, where $M$ is a large constant, $M\gg 2$:
\begin{center}
\begin{tabular}{lcccc}
    & &   &  &  \\
\hline \multicolumn{1}{|c}{$\vabs_A$} & \multicolumn{1}{|c}{\cellcolor{myGreen!25}$M$} & \multicolumn{1}{|c}{\cellcolor{myGreen!25}$M$} & \multicolumn{1}{|c}1 & \multicolumn{1}{|c|}1    \\
\hline \multicolumn{1}{|c}{$\vabs_B$}  & \multicolumn{1}{|c}1 & \multicolumn{1}{|c}1 & \multicolumn{1}{|c}{\cellcolor{myBlue!25}$M$} & \multicolumn{1}{|c|}{\cellcolor{myBlue!25}$M$}  \\
\hline
\end{tabular}
\quad\quad
\begin{tabular}{lcccccc}
    & &   &  & & $\blacktriangledown$ & $\blacktriangledown$ \\
\hline \multicolumn{1}{|c}{$\vabs_A$} & \multicolumn{1}{|c}{\cellcolor{myGreen!25}$M$} &
\multicolumn{1}{|c}{\cellcolor{myGreen!25}$M$}&
\multicolumn{1}{|c}{\cellcolor{myGreen!25}1} &
\multicolumn{1}{|c}1 &
\multicolumn{1}{|c}{$M$} &
\multicolumn{1}{|c|}{$M$}
\\
\hline \multicolumn{1}{|c}{$\vabs_B$}  & \multicolumn{1}{|c}1 &
\multicolumn{1}{|c}1 &
\multicolumn{1}{|c}{$M$} & \multicolumn{1}{|c}{\cellcolor{myBlue!25}$M$} &  \multicolumn{1}{|c}{\cellcolor{myBlue!25}1} & \multicolumn{1}{|c|}{\cellcolor{myBlue!25}1} \\
\hline
\end{tabular}

\end{center}

In the smaller cake, the unique max-relative-equitable division gives the two leftmost slices to Alice and the two rightmost slices to Bob. The relative-equitable value is $M/(M+2)$. However, in the larger cake the unique max-relative-equitable division is attained by cutting exactly in the middle, decreasing the relative-equitable value to $1/2$. While Alice gains from the division and her (absolute) value increases by $1$, Bob loses since his absolute value drops from $2 M$ to $M+2$.
\end{example}

The following theorem summarizes the properties of the relative-equitable rule:
\begin{theorem}
\label{thm:relative-equitable}
The relative-equitable rule is weakly-Pareto-optimal and population-monotonic and proportional, but not resource-monotonic.
\end{theorem}

\subsection{\textbf{Absolute-equitable rule: WPO+PM+RM}}  \label{sub:absolute-equitable}
As shown by Example \ref{exm:relative-equitable-not-rm}, the relative-equitable rule is not RM since the relative value of some agents is made smaller when the cake becomes larger. This may imply that, if we use \emph{absolute} instead of relative values, we can get resource-monotonicity.

Fortunately, almost all definitions, examples, procedures, lemmata and proofs from the previous subsection can easily be adapted to absolute values by just replacing ``relative'' with ``absolute''; the only exception is Lemma \ref{lem:relative-equitable}.

\begin{theorem}
\label{thm:absolute-equitable}
The absolute-equitable rule is
weakly-Pareto-optimal and population-monotonic and resource-monotonic, but not proportional.
\end{theorem}
\begin{proof}
WPO holds by Lemma \ref{lemma:equitable-is-wpo} and PM by Lemma \ref{lemma:equitable-is-pm}, replacing ``relative'' by ``absolute''.

The proof of RM is essentially the same as the proof of Lemma \ref{lemma:equitable-is-pm}: the cake enlargement can be treated as a piece that was acquired from an agent who left the scene. \footnote{Note that this argument is not true for the relative-equitable-connected rule. When the cake grows, while the absolute value of all agents weakly increases, the relative value of some agents may decrease. Hence, the relative-equitable value in the enlarged cake might be smaller than in the original cake, and this may make some agents worse-off. See Example \ref{exm:relative-equitable-not-rm}.}

To see that the rule is not PROP, suppose that Alice values the entire cake as $1$ and Bob values the entire cake as $M\gg 2$. Then, any absolute-equitable division must give Bob at most $1/M \ll 1/2$ of his value.
\end{proof}

\subsection{\textbf{Rightmost-mark rule: WPO+PROP+RM}} \label{sub:right-mark}
In this section we present a resource-monotonic procedure that produces an envy-free (hence proportional) division of the whole cake for 2 agents, giving each agent a connected piece.

\begin{figure}[h!]
	\begin{center}
		\includegraphics[width=0.8\columnwidth]{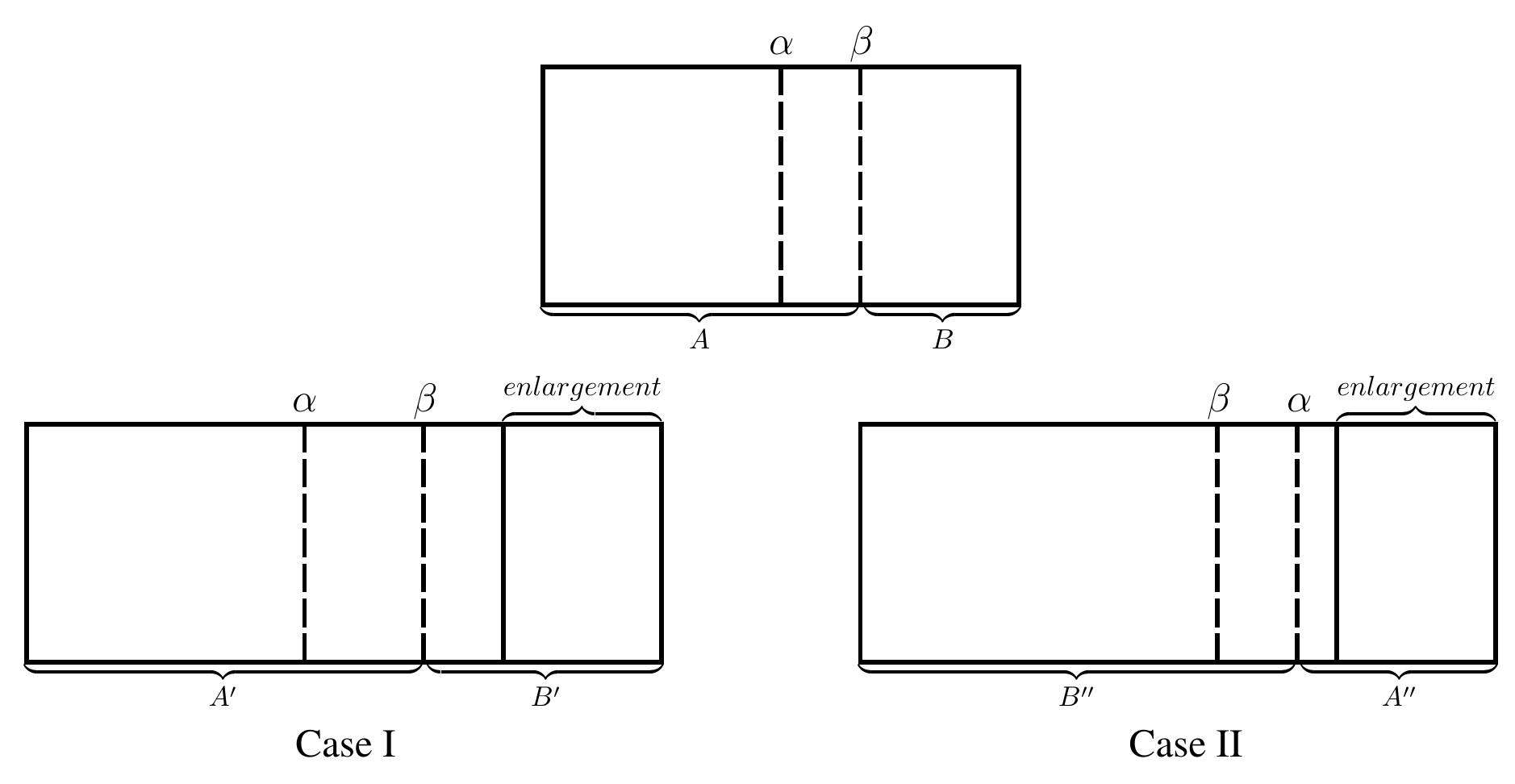}
		\caption{Illustration of the rightmost-mark division rule. Alice's cut mark is denoted by $\alpha$, while Bob's cut mark by $\beta$.}\label{fig:rightmost_mark}
	\end{center}
\end{figure}

The procedure is called the \emph{rightmost-mark rule} and consists of the following steps.

\begin{itemize}
	\item{Ask both agents to make a mark which cuts the cake in half according to their own valuation. If more than one point satisfies this criterion, i.e.\ the middle of the cake is worthless to one of the agents, take the rightmost such point.}
	\item{Cut the cake at the rightmost mark and give the slice on the right to the agent who made the mark.}
	\item{The remaining part is given to the other agent.}
\end{itemize}

\begin{theorem}
For two agents with connected utilities, the rightmost-mark procedure is envy-free, proportional, weakly Pareto-optimal and resource-monotonic.
\end{theorem}
\begin{proof}
EF and PROP are obvious.

To prove WPO, suppose w.l.o.g. that Bob made the rightmost mark. Suppose we want to give Bob a piece worth strictly more than his current utility of $1/2$. This can be done in two ways. One way is to keep Bob at the right side and move the division line leftwards; this necessarily does not increase Alice's utility. The other way is to switch between Alice and Bob. But then, if Bob's utility is to be improved, he must receive at least Alice's current share (which is worth for him $1/2$). This leaves at most $1/2$ to Alice. Hence, there is no division in which the utilities of both agents are strictly higher.

We now prove that the procedure is RM. Since the rule is single-valued, it is sufficient to prove upwards-RM.

Suppose w.l.o.g. that Bob made the rightmost mark on the smaller cake (see Fig. \ref{fig:rightmost_mark}). Thus, Bob obtained the piece marked with B, which is worth exactly half for him. Alice received the part marked with A, which is worth at least half for her. When the cake is enlarged two cases are possible: The order of the cut marks made by the agents remains the same or gets reversed. In the first case, Bob still receives the rightmost cake (marked with B'). Since it still worth for him half of the cake, and since the cake is enlarged, he is not worse off. Neither is Alice, who receives a piece that contains her original share.

In the second case, Alice receives the rightmost piece A''. Note that she believes that the pieces A'' and B'' represent the same value, and B'' contains A, her original piece. Thus, she is not worse off. Similarly, Bob evaluates A and B the same, and he received B'' which contains A, thus he is not worse off either.
\end{proof}

Now we show that, in the special case in which the value-densities of both agents are strictly-positive, the rightmost-mark division rule is the \emph{only} rule which is PROP+WPO+RM. First we need a lemma.

\begin{lemma}
\label{lem:prop-structure}
Suppose there are $n=2$ agents, Alice and Bob, with strictly-positive valuations, and half-points $h_A,h_B$ respectively. Then, any PROP+WPO division rule that allocates connected pieces must:

(1) Cut the cake in the closed interval between $h_A$ and $h_B$ (that is, $[h_A,h_B]$ if $h_A\leq h_B$ or $[h_B,h_A]$ if $h_B\leq h_A$);

(2) Allocate the leftmost piece to the agent with the leftmost half-point (that is, Alice if $h_A\leq h_B$ or Bob if $h_B\leq h_A$).
\end{lemma}
\begin{proof}
If $h_A=h_B$ then obviously the only PROP allocation is to cut at that point and give a half to each agent.

W.l.o.g, we now assume that $h_A< h_B$. If the cake is cut at $x<h_A$, then the piece to the left of x is worth less than 1/2 to both agents, so it cannot be given to any of them. Similarly, if the cake is cut at $x>h_B$, then the piece to the right of x is worth less than 1/2 to both agents. Hence, the cake must be cut at $x\in [h_A,h_B]$.

Since $h_A<h_B$, we must have either $h_A<x$ or $x<h_B$ (or both). In the former case, the piece to the right of $x$ is worth less than $1/2$ to Alice; in the latter case, the piece to the left of $x$ is worth less than $1/2$ to Bob. So in both cases, Alice must get the left and Bob must get the right.

\end{proof}

\begin{theorem}\label{thm:prop-rm-half}
For $n=2$ agents with strictly-positive valuations, any PROP+WPO+RM division rule that allocates connected pieces must allocate the rightmost agent a relative utility of exactly 1/2.
\end{theorem}
\begin{proof}
If $h_A=h_B$ then obviously both agents must receive exactly 1/2. W.l.o.g, we now assume that $h_A<h_B$, so the rightmost agent is Bob. We normalize the agents' valuations to 2. So the cake looks like the following (where $a\in[0,1]$ and $b\in[0,1]$ are some constants):
\begin{center}
\begin{tabular}{lcccc}
& $[0,h_A)$ & $[h_A,h_B)$ & $[h_B,1]$  \\
\hline
\c{$\vabs_A$} & \c{$1$} & \c{$a$} & \c{$1-a$} \\
\hline
\c{$\vabs_B$} & \c{$b$} & \c{$1-b$} & \c{$1$} \\
\hline
\end{tabular}
\end{center}
We claim that a PROP+WPO+RM algorithm must give Bob a value of at most 1. Suppose by contradiction that Bob's value is $1+2 d$, where $d>0$ is a constant, $d\in(0,1-b)$. Now, the cake grows as follows:
\begin{center}
\begin{tabular}{lccccc}
& &   & & $\blacktriangledown$ \\
\hline
\c{$\vabs_A$} & \c{$1$} & \c{$a$} & \c{$1-a$} & \c{$2a$} \\
\hline
\c{$\vabs_B$} & \c{$b$} & \c{$1-b$} & \c{$1$} & \c{$d$}  \\
\hline
\end{tabular}
\end{center}
In the extended cake, $h_A$ moves rightwards and is located exactly between slices \#2 and \#3. $h_B$ also moves slightly rightwards and is now located inside slice \#3. By Lemma \ref{lem:prop-structure}, the cake is cut at or to the right of the new $h_A$ and Bob receives the rightmost piece. Hence, Bob's new value is at most $1+d$ - in contradiction to RM.
\end{proof}

Theorem \ref{thm:prop-rm-half} implies that, when the value-densities are strictly-positive, the rightmost-mark rule is the only rule that satisfies PROP+WPO+RM with connected utilities.

\section{Conclusion and Future Work} \label{sec:conclusion}
We studied monotonicity properties in combination with the classical axioms of proportionality and Pareto-optimality. Table \ref{tab:summary} summarizes the properties of the various division rules. Most properties are proved in the paper body, except the WPO properties of the classic protocols, which are proved in Appendix \ref{sec:WPO}.

\def\y{\c{\yy}}  
\def\n{\c{\nn}}   
\def\N{\c{\NN}}  
\def\Y{\c{\YY}}  

\def\yy{\textcolor[RGB]{0,200,0}{Yes}}  
\def\nn{\textcolor[RGB]{200,0,0}{No}}   
\def\NN{\textcolor[RGB]{0,0,200}{No*}}  
\def\YY{\textcolor[RGB]{179,109,86}{Y.c.u.}}  
\def\upw{\textcolor[RGB]{190,125,219}{Upw}} 
\def\c#1{\multicolumn{1}{|c|}{#1}}  
\def\h#1{\c{\footnotesize #1}}  

\begin{table}[h!]
\begin{center}
\footnotesize
\begin{tabular}{lcccccccc}
\hline

\c{\textbf{Connected rules}} & \c{$n$} &  \h{CON}& \h{EF}&\h{PROP}&\h{PO}&\h{WPO}&\h{RM}&\h{PM} \\ \hline

\c{exact-proportional} & \c{Any}  &\y & \n & \y & \N & \n & \y & \y  \\ \hline

\c{absolute-equitable} & \c{Any}  &\y & \n & \n & \n & \YY & \y & \y  \\ \hline

\c{relative-equitable} & \c{Any}& \y & \n & \y & \N & \YY & \n & \y  \\ \hline

\c{rightmost-mark} & \c{2} & \y &\y & \y & \N & \YY & \y & \n  \\ \hline

 & &  & & & & & \\ \hline

\c{\textbf{Classic rules}}       & \c{$n$}  & \h{CON}&\h{EF} & \h{PROP} &\h{PO} &\h{WPO} &\h{RM} &\h{PM}  \\ \hline

\c{Banach-Knaster}       & \c{Any}& \y & \n & \y & \n & \n & \n & \n  \\ \hline

\c{Cut and Choose}       & \c{2}& \y  & \y & \y & \n & \YY & \n & \n  \\ \hline

\c{Dubins-Spanier}       & \c{Any}& \y  & \n & \y & \n & \n & \n & \n  \\ \hline

\c{Even-Paz}             & \c{Any}& \y  & \n & \y & \n & \n & \n & \n  \\ \hline

\c{Fink}                 & \c{Any}& \n  & \n & \y & \n & \n & \n & \c{\upw}  \\ \hline

\c{Selfridge-Conway}     & \c{3}& \n  & \y & \y & \n & \n & \n & \n  \\ \hline

\end{tabular}

\protect\caption{\label{tab:summary}Properties of division rules presented in this paper. \nn{} means that the property is not satisfied by the rule, while \NN{} means that the property cannot be satisfied by \emph{any} rule satisfying the other properties marked with \yy{} in the same line. In the WPO column \YY{} stands for 'Yes for connected utilities'. In case of the Fink rule, \upw{} indicates that -- with some additional adjustments -- the rule is upwards-PM.
}
\end{center}
\end{table}

Each of our connected division rules satisfies three of the four properties \{PROP,WPO,RM,PM\}. Thus, the divider has to choose whether to give up proportionality (PROP) or efficiency (WPO) or give up one of the monotonicity properties. We are still missing a rule that satisfies PROP+WPO+RM for three or more agents, as well as a rule that satisfies PROP+WPO+RM+PM for two or more agents. Additionally, combining envy-freeness with monotonicity for three or more agents looks like a fairly challenging task.

In this paper we ignored strategic considerations and assumed that all agents truthfully report their valuations. An interesting future
research topic is how to ensure monotonicity in truthful division procedures.

Finally, our procedure for equitable division uses moving knives and thus it is not discrete.  Recently, \citet{Cechlarova2011Near} and \citet{Cechlarova2012Computability} presented discrete procedures that attain approximately-equitable connected divisions. A division rule based on such procedures naturally attains approximate versions of proportionality and monotonicity. Further development of this idea is deferred to future work.

\section{Acknowledgments}
The idea of this paper was born in the COST Summer School on Fair Division in Grenoble, 7/2015 (FairDiv-15). We are grateful to COST and the conference organizers for the wonderful opportunity to meet with fellow researchers from around the globe. In particular, we are grateful to Ioannis Caragiannis, Ulle Endriss and Christian Klamler for sharing their insights on cake-cutting. We are also thankful to Marcus Berliant, Shiri Alon-Eron, Christian Blatter and Ilan Nehama for their very helpful comments.

The authors acknowledge the support of the `Momentum' Programme (LP-004/2010) of the Hungarian Academy of Sciences, the Pallas Athene Domus Scientiae Foundation, the OTKA grants K108383 and K109354, 
the ISF grant 1083/13, the Doctoral Fellowships of Excellence Program, the Wolfson Chair and the Mordecai and Monique Katz Graduate Fellowship Program at Bar-Ilan University.  In addition Sziklai was supported by the \'UNKP-16-4-I. New National Excellence Program of the Ministry of Human Capacities.

\appendix

\section{Existence of equitable-connected divisions}
\label{sub:max-eq-proof}
The proof uses the \emph{Borsuk-Ulam theorem}\footnote{Independently and contemporaneously to our work, \citet{Cheze2017} came up with a similar idea.}. It is about functions defined on spheres. Define the sphere $S^{n-1}$ as the set of points $(x_1,\dots,x_n)$ satisfying:
$|x_1|+\cdots+|x_n| = 1$ (it is a sphere in the $\ell_1$ metric).

\begin{theorem*}[Borsuk-Ulam]
Let $f_i$, for $i=1,\dots,n-1$, be real-valued functions of $n$ variables, that are continuous on the sphere $S^{n-1}$ .

Then, there exists a point on the sphere, $X^*=(x_1,\dots,x_n)\in S^{n-1}$,
such that for all $i$: $f_i(X^*) = f_i(-X^*)$.
\end{theorem*}

Assume that the cake is the interval $[0,1]$. Each point $(x_1,\dots,x_n) \in S^{n-1}$ corresponds to a partition of the cake to $n$ intervals, marked $X_1,\dots,X_n$, such that the length of interval $X_i$ (the $i$-th interval from the left) is $|x_i|$. Note that each partition corresponds to many points which differ in the signs of some or all of the coordinates (this representation of the partition space was introduced by \cite{Alon1986BorsukUlam} and used e.g.\ by \cite{Simmons2003}).

Suppose w.l.o.g. that the players are ordered from $1$ to $n$, so that player $i$ receives the piece $X_i$. For every point $X=(x_1,\dots,x_n)\in S^{n-1}$ and for every $i\in 1,\dots,n-1$, define the function $f_i(X)$ as follows:
\[
    f_i(X) = \sign(x_i)\cdot \vrel_i(X_i) - \sign(x_{i+1})\cdot \vrel_{i+1}(X_{i+1})
\]

Note that when $x_i=0$, interval $X_i$ is empty so $\vrel_i(X_i)=0$. Hence, the functions $f_i$ are continuous on $S^{n-1}$.

Hence, by the Borsuk-Ulam theorem, there exists a point $X^*$ on $S^{n-1}$ such that for all $i$: $f_i(X^*) = f_i(-X^*)$. By definition of the $f_i$, the cake division that corresponds to $X^*$ necessarily satisfies:

\[
     \sign(x_i)\cdot \vrel_i(X^*_i) = \sign(x_{i+1})\cdot \vrel_{i+1}(X^*_{i+1})
\]

This is possible only if for all $i\in(1,\dots,n-1)$:

\[
     \vrel_i(X^*_i) = \vrel_{i+1}(X^*_{i+1})
\]
Hence the division $X^*$ is equitable.

\section{Weak Pareto-optimality of classic protocols}
\label{sec:WPO}
It is a well-known fact that the classic protocols that we discuss here are not Pareto-optimal. Now we show that -- with the exception of Cut and Choose -- they are not even weakly Pareto-optimal. In the following example the Banach-Knaster, Dubins-Spanier and Even-Paz protocols coincide.

\begin{center}
    \begin{tabular}{lcccccc}
    \hline
    \multicolumn{1}{|c}{$\vabs_A$} & \multicolumn{1}{|c}{\cellcolor{myGreen!25}2} & \multicolumn{1}{|c}0 & \multicolumn{1}{|c}0 & \multicolumn{1}{|c}0 & \multicolumn{1}{|c}0 & \multicolumn{1}{|c|}4 \\
    \hline
     \multicolumn{1}{|c}{$\vabs_B$} &  \multicolumn{1}{|c}2 &  \multicolumn{1}{|c}{\cellcolor{myBlue!25}3} &  \multicolumn{1}{|c}{\cellcolor{myBlue!25}1} &  \multicolumn{1}{|c}{\cellcolor{myBlue!25}1} & \multicolumn{1}{|c}5 & \multicolumn{1}{|c|}0\\
    \hline
    \multicolumn{1}{|c}{$\vabs_C$} & \multicolumn{1}{|c}2 & \multicolumn{1}{|c}3& \multicolumn{1}{|c}1 &   \multicolumn{1}{|c}{1} & \multicolumn{1}{|c}{\cellcolor{myOrchid!25}5} & \multicolumn{1}{|c|}{\cellcolor{myOrchid!25}0}\\
    \hline
  \end{tabular}
    \quad\quad
    \begin{tabular}{lcccccc}
    \hline
    \multicolumn{1}{|c}{$\vabs_A$} & \multicolumn{1}{|c}{2} & \multicolumn{1}{|c}0 & \multicolumn{1}{|c}0 & \multicolumn{1}{|c}0 & \multicolumn{1}{|c}0 & \multicolumn{1}{|c|}{\cellcolor{myGreen!25}4} \\
    \hline
     \multicolumn{1}{|c}{$\vabs_B$} &  \multicolumn{1}{|c}{\cellcolor{myBlue!25}2} &  \multicolumn{1}{|c}{\cellcolor{myBlue!25}3} &  \multicolumn{1}{|c}{\cellcolor{myBlue!25}1} &  \multicolumn{1}{|c}1 & \multicolumn{1}{|c}5 & \multicolumn{1}{|c|}0\\
    \hline
    \multicolumn{1}{|c}{$\vabs_C$} & \multicolumn{1}{|c}2 & \multicolumn{1}{|c}3& \multicolumn{1}{|c}1 &   \multicolumn{1}{|c}{\cellcolor{myOrchid!25}1} & \multicolumn{1}{|c}{\cellcolor{myOrchid!25}5} & \multicolumn{1}{|c|}{0}\\
    \hline
  \end{tabular}
\end{center}

Alice gets the first slice as it composes $1/3$ of her cake value, while Bob and Carl perform a Cut and Choose on the rest of the cake. The table on the right presents an alternative allocation (still with connected pieces) which is strictly better for all agents.

The Fink method is not contiguous, hence we can obtain an improvement by composing the pieces from more slices.

\begin{table}[h!]
  \centering
    \begin{tabular}{lcccc}
         &  &    & & \\
    \hline
    \multicolumn{1}{|c}{$\vabs_A$} & \multicolumn{1}{|c}{\cellcolor{myGreen!25}0} & \multicolumn{1}{|c}{\cellcolor{myGreen!25}3} & \multicolumn{1}{|c}2  & \multicolumn{1}{|c|}1 \\
    \hline
    \multicolumn{1}{|c}{$\vabs_B$} & \multicolumn{1}{|c}2 & \multicolumn{1}{|c}1 & \multicolumn{1}{|c}{\cellcolor{myBlue!25}2}  & \multicolumn{1}{|c|}{\cellcolor{myBlue!25}$1+\epsilon$} \\
    \hline
  \end{tabular}
  \quad\quad
      \begin{tabular}{lcccc}
           &  &    &  &  \\
    \hline
    \multicolumn{1}{|c}{$\vabs_A$} & \multicolumn{1}{|c}{0} & \multicolumn{1}{|c}{\cellcolor{myGreen!25}3} & \multicolumn{1}{|c}{2}  & \multicolumn{1}{|c|}{\cellcolor{myGreen!25}1} \\
    \hline
    \multicolumn{1}{|c}{$\vabs_B$} & \multicolumn{1}{|c}{\cellcolor{myBlue!25}2} & \multicolumn{1}{|c}1 & \multicolumn{1}{|c}{\cellcolor{myBlue!25}2}  & \multicolumn{1}{|c|}{$1+\epsilon$} \\
    \hline
  \end{tabular}
\end{table}

For two agents the Fink method proceeds as the Cut and Choose: Alice cuts in the middle and Bob chooses the piece on the right. The second cake is a Pareto-improvement with four slices where every agent is strictly better off.

The same example shows that the Cut and Choose is not WPO for additive utilities. However, it is a contiguous protocol, thus it makes sense to investigate whether we could improve it with connected pieces.

\begin{lemma}\label{CC-nowpo}
The Cut and Choose method is weakly Pareto-optimal whenever the agents have connected utilities.
\end{lemma}

\begin{proof}
Suppose Alice cuts at point $x \in [0,c]$ and Bob chooses. The cake is divided into the left piece ($L$) and the right piece ($R$). Without loss of generality we may assume that Alice receives the left piece. Alice's utility can be improved in only two ways: (a) Cut at some point $x'>x$ and give the left piece to Alice. But then Bob would receive $R' \subset R$, which obviously cannot be strictly better than $R$. (b) Cut at some point $x''<x$ and give Alice the right piece $R'' \supset R$. Then Bob receives $L'' \subset L$. Since Bob preferred $R$ to $L$ he did not gain utility by this swap. So it is impossible to strictly increase the utility of Alice and Bob simultaneously.
\end{proof}

Finally we show that the Selfridge-Conway does not satisfy WPO either. Consider the following cake.

\begin{table}[h!]
  \centering
    \begin{tabular}{lcccccc}
    \hline
    \multicolumn{1}{|c}{$\vabs_A$} & \multicolumn{1}{|c}{3} & \multicolumn{1}{|c?{0.5mm}}1 & \multicolumn{1}{c}3 & \multicolumn{1}{|c?{0.5mm}}1 & \multicolumn{1}{c}{\cellcolor{myGreen!25}2} & \multicolumn{1}{|c|}{\cellcolor{myGreen!25}2} \\
    \hline
     \multicolumn{1}{|c}{$\vabs_B$} &  \multicolumn{1}{|c}1 &  \multicolumn{1}{|c?{0.5mm}}{3} &  \multicolumn{1}{c}{\cellcolor{myBlue!25}1} &  \multicolumn{1}{|c?{0.5mm}}{\cellcolor{myBlue!25}3} & \multicolumn{1}{c}1 & \multicolumn{1}{|c|}2\\
    \hline
    \multicolumn{1}{|c}{$\vabs_C$} & \multicolumn{1}{|c}{\cellcolor{myOrchid!25}4} & \multicolumn{1}{|c?{0.5mm}}{\cellcolor{myOrchid!25}0}& \multicolumn{1}{c}0 &   \multicolumn{1}{|c?{0.5mm}}{0} & \multicolumn{1}{c}{0} & \multicolumn{1}{|c|}{3}\\
    \hline
  \end{tabular}
    \quad\quad
      \begin{tabular}{lcccccc}
    \hline
    \multicolumn{1}{|c}{$\vabs_A$} & \multicolumn{1}{|c}{3} & \multicolumn{1}{|c}1 & \multicolumn{1}{|c}{\cellcolor{myGreen!25}3} & \multicolumn{1}{|c}1 & \multicolumn{1}{|c}{\cellcolor{myGreen!25}2} & \multicolumn{1}{|c|}{2} \\
    \hline
     \multicolumn{1}{|c}{$\vabs_B$} &  \multicolumn{1}{|c}1 &  \multicolumn{1}{|c}{\cellcolor{myBlue!25}3} &  \multicolumn{1}{|c}{1} &  \multicolumn{1}{|c}{\cellcolor{myBlue!25}3} & \multicolumn{1}{|c}1 & \multicolumn{1}{|c|}2\\
    \hline
    \multicolumn{1}{|c}{$\vabs_C$} & \multicolumn{1}{|c}{\cellcolor{myOrchid!25}4} & \multicolumn{1}{|c}{0}& \multicolumn{1}{|c}0 &   \multicolumn{1}{|c}{0} & \multicolumn{1}{|c}{0} & \multicolumn{1}{|c|}{\cellcolor{myOrchid!25}3}\\
    \hline
  \end{tabular}
\end{table}

Alice cuts the cake into three equal parts. Since the two most valuable slices have equal value for Bob he passes. The agents choose in the Carl-Bob-Alice order and obtain the pieces shown by the table on the left. However, as the table on the right shows, there is an allocation where every agent is strictly better off.

\newpage

\small{
\bibliography{CCM,FairAndSquare}

\begin{thebibliography}{47}
\expandafter\ifx\csname natexlab\endcsname\relax\def\natexlab#1{#1}\fi
\expandafter\ifx\csname url\endcsname\relax
  \def\url#1{\texttt{#1}}\fi
\expandafter\ifx\csname urlprefix\endcsname\relax\def\urlprefix{URL }\fi
\providecommand{\eprint}[2][]{\url{#2}}
\providecommand{\bibinfo}[2]{#2}
\ifx\xfnm\relax \def\xfnm[#1]{\unskip,\space#1}\fi
\bibitem[{Alon(1987)}]{Alon_1987}
\bibinfo{author}{Alon, N.}, \bibinfo{year}{1987}.
\newblock \bibinfo{title}{{Splitting necklaces}}.
\newblock \bibinfo{journal}{Advances in Mathematics} \bibinfo{volume}{63},
  \bibinfo{pages}{247--253}.
\bibitem[{Alon and West(1986)}]{Alon1986BorsukUlam}
\bibinfo{author}{Alon, N.}, \bibinfo{author}{West, D.B.}, \bibinfo{year}{1986}.
\newblock \bibinfo{title}{{The Borsuk-Ulam theorem and bisection of
  necklaces}}.
\newblock \bibinfo{journal}{Proceedings of the American Mathematical Society}
  \bibinfo{volume}{98}, \bibinfo{pages}{623--628}.
\bibitem[{Arzi et~al.(2011)Arzi, Aumann and Dombb}]{Arzi2011}
\bibinfo{author}{Arzi, O.}, \bibinfo{author}{Aumann, Y.},
  \bibinfo{author}{Dombb, Y.}, \bibinfo{year}{2011}.
\newblock \bibinfo{title}{Throw one's cake - and eat it too}, in:
  \bibinfo{editor}{Persiano, G.} (Ed.), \bibinfo{booktitle}{Algorithmic Game
  Theory}. \bibinfo{publisher}{Springer Berlin Heidelberg}. volume
  \bibinfo{volume}{6982} of \textit{\bibinfo{series}{Lecture Notes in Computer
  Science}}, pp. \bibinfo{pages}{69--80}.
\bibitem[{Aumann and Dombb(2010)}]{Aumann2010Efficiency}
\bibinfo{author}{Aumann, Y.}, \bibinfo{author}{Dombb, Y.},
  \bibinfo{year}{2010}.
\newblock \bibinfo{title}{{The Efficiency of Fair Division with Connected
  Pieces}}.
\newblock \bibinfo{journal}{Web, Internet and Network Economics}
  \bibinfo{volume}{6484}, \bibinfo{pages}{26--37}.
\bibitem[{Aziz and Ye(2014)}]{Aziz2014Cake}
\bibinfo{author}{Aziz, H.}, \bibinfo{author}{Ye, C.}, \bibinfo{year}{2014}.
\newblock \bibinfo{title}{{Cake Cutting Algorithms for Piecewise Constant and
  Piecewise Uniform Valuations}}, in: \bibinfo{editor}{Liu, T.Y.},
  \bibinfo{editor}{Qi, Q.}, \bibinfo{editor}{Ye, Y.} (Eds.),
  \bibinfo{booktitle}{Web and Internet Economics}. \bibinfo{publisher}{Springer
  International Publishing}. volume \bibinfo{volume}{8877} of
  \textit{\bibinfo{series}{Lecture Notes in Computer Science}}, pp.
  \bibinfo{pages}{1--14}.
\bibitem[{Balinski and Young(1975)}]{Balinski1975}
\bibinfo{author}{Balinski, M.}, \bibinfo{author}{Young, H.P.},
  \bibinfo{year}{1975}.
\newblock \bibinfo{title}{{The quota method of apportionment}}.
\newblock \bibinfo{journal}{American Mathematical Monthly}
  \bibinfo{volume}{82}, \bibinfo{pages}{701--730}.
\bibitem[{Balinski and Young(1982)}]{Balinski1982}
\bibinfo{author}{Balinski, M.}, \bibinfo{author}{Young, H.P.},
  \bibinfo{year}{1982}.
\newblock \bibinfo{title}{{Fair Representation: Meeting the Ideal of One Man,
  One Vote}}.
\newblock \bibinfo{publisher}{Yale University Press}, \bibinfo{address}{New
  Haven}.
\bibitem[{Berliant and Raa(1988)}]{Berliant1988Foundation}
\bibinfo{author}{Berliant, M.}, \bibinfo{author}{Raa, T.},
  \bibinfo{year}{1988}.
\newblock \bibinfo{title}{{A foundation of location theory: Consumer
  preferences and demand}}.
\newblock \bibinfo{journal}{Journal of Economic Theory} \bibinfo{volume}{44},
  \bibinfo{pages}{336--353}.
\bibitem[{Berliant et~al.(1992)Berliant, Thomson and Dunz}]{Berliant1992Fair}
\bibinfo{author}{Berliant, M.}, \bibinfo{author}{Thomson, W.},
  \bibinfo{author}{Dunz, K.}, \bibinfo{year}{1992}.
\newblock \bibinfo{title}{{On the fair division of a heterogeneous commodity}}.
\newblock \bibinfo{journal}{Journal of Mathematical Economics}
  \bibinfo{volume}{21}, \bibinfo{pages}{201--216}.
\bibitem[{Brams and Taylor(1996)}]{Brams1996}
\bibinfo{author}{Brams, S.}, \bibinfo{author}{Taylor, A.},
  \bibinfo{year}{1996}.
\newblock \bibinfo{title}{{Fair Division: From Cake Cutting to Dispute
  Resolution}}.
\newblock \bibinfo{publisher}{Cambridge University Press},
  \bibinfo{address}{Cambridge UK}.
\bibitem[{Brams et~al.(2006)Brams, Jones and Klamler}]{BramsJonesKlamler2006}
\bibinfo{author}{Brams, S.J.}, \bibinfo{author}{Jones, M.A.},
  \bibinfo{author}{Klamler, C.}, \bibinfo{year}{2006}.
\newblock \bibinfo{title}{{Better Ways to Cut a Cake}}.
\newblock \bibinfo{journal}{Notices of the AMS} \bibinfo{volume}{53},
  \bibinfo{pages}{1314--1321}.
\bibitem[{Brams et~al.(2007)Brams, Jones and Klamler}]{BramsJonesKlamler2007}
\bibinfo{author}{Brams, S.J.}, \bibinfo{author}{Jones, M.A.},
  \bibinfo{author}{Klamler, C.}, \bibinfo{year}{2007}.
\newblock \bibinfo{title}{{Better Ways to Cut a Cake - Revisited}}.
\newblock \bibinfo{type}{mimeo}. Dagstuhl Seminar Proceedings 07261.
\bibitem[{Br\^{a}nzei and Miltersen(2015)}]{branzei2015dictatorship}
\bibinfo{author}{Br\^{a}nzei, S.}, \bibinfo{author}{Miltersen, P.B.},
  \bibinfo{year}{2015}.
\newblock \bibinfo{title}{{A Dictatorship Theorem for Cake Cutting}}, in:
  \bibinfo{booktitle}{Proceedings of the 24th International Conference on
  Artificial Intelligence}, \bibinfo{publisher}{AAAI Press}. pp.
  \bibinfo{pages}{482--488}.
\bibitem[{Br\^{a}nzei and Nisan(2017)}]{Branzei2017Query}
\bibinfo{author}{Br\^{a}nzei, S.}, \bibinfo{author}{Nisan, N.},
  \bibinfo{year}{2017}.
\newblock \bibinfo{title}{{The Query Complexity of Cake Cutting}}.
\newblock \bibinfo{note}{Preprint https://arxiv.org/abs/1705.02946}.
\bibitem[{Calleja et~al.(2012)Calleja, Rafels and Tijs}]{Calleja2012}
\bibinfo{author}{Calleja, P.}, \bibinfo{author}{Rafels, C.},
  \bibinfo{author}{Tijs, S.}, \bibinfo{year}{2012}.
\newblock \bibinfo{title}{Aggregate monotonic stable single-valued solutions
  for cooperative games}.
\newblock \bibinfo{journal}{International Journal of Game Theory}
  \bibinfo{volume}{41}, \bibinfo{pages}{899--913}.
\bibitem[{Caulfield(2008)}]{caulfield2008apportioning}
\bibinfo{author}{Caulfield, M.J.}, \bibinfo{year}{2008}.
\newblock \bibinfo{title}{Apportioning representatives in the united states
  congress}.
\newblock \bibinfo{journal}{AMC} \bibinfo{volume}{10}, \bibinfo{pages}{12}.
\bibitem[{Cechl\'{a}rov\'{a} et~al.(2013)Cechl\'{a}rov\'{a}, Dobo\v{s} and
  Pill\'{a}rov\'{a}}]{Cechlarova2013Existence}
\bibinfo{author}{Cechl\'{a}rov\'{a}, K.}, \bibinfo{author}{Dobo\v{s}, J.},
  \bibinfo{author}{Pill\'{a}rov\'{a}, E.}, \bibinfo{year}{2013}.
\newblock \bibinfo{title}{{On the existence of equitable cake divisions}}.
\newblock \bibinfo{journal}{Information Sciences} \bibinfo{volume}{228},
  \bibinfo{pages}{239--245}.
\bibitem[{Cechl\'{a}rov\'{a} and Pill\'{a}rov\'{a}(2011)}]{Cechlarova2011Near}
\bibinfo{author}{Cechl\'{a}rov\'{a}, K.}, \bibinfo{author}{Pill\'{a}rov\'{a},
  E.}, \bibinfo{year}{2011}.
\newblock \bibinfo{title}{{A near equitable 2-person cake cutting algorithm}}.
\newblock \bibinfo{journal}{Optimization} \bibinfo{volume}{61},
  \bibinfo{pages}{1321--1330}.
\bibitem[{Cechl\'{a}rov\'{a} and
  Pill\'{a}rov\'{a}(2012)}]{Cechlarova2012Computability}
\bibinfo{author}{Cechl\'{a}rov\'{a}, K.}, \bibinfo{author}{Pill\'{a}rov\'{a},
  E.}, \bibinfo{year}{2012}.
\newblock \bibinfo{title}{{On the computability of equitable divisions}}.
\newblock \bibinfo{journal}{Discrete Optimization} \bibinfo{volume}{9},
  \bibinfo{pages}{249--257}.
\bibitem[{Chambers(2005a)}]{Chambers_2005}
\bibinfo{author}{Chambers, C.P.}, \bibinfo{year}{2005}a.
\newblock \bibinfo{title}{{Allocation rules for land division}}.
\newblock \bibinfo{journal}{Journal of Economic Theory} \bibinfo{volume}{121},
  \bibinfo{pages}{236--258}.
\bibitem[{Chambers(2005b)}]{Chambers2005Allocation}
\bibinfo{author}{Chambers, C.P.}, \bibinfo{year}{2005}b.
\newblock \bibinfo{title}{{Allocation rules for land division}}.
\newblock \bibinfo{journal}{Journal of Economic Theory} \bibinfo{volume}{121},
  \bibinfo{pages}{236--258}.
\bibitem[{Ch\'eze(2017)}]{Cheze2017}
\bibinfo{author}{Ch\'eze, G.}, \bibinfo{year}{2017}.
\newblock \bibinfo{title}{{Existence of a simple and equitable fair division: A
  short proof}}.
\newblock \bibinfo{journal}{Mathematical Social Sciences}
  \bibinfo{volume}{Short communication}.
\bibitem[{Dall'Aglio(2001)}]{DallAglio2001DubinsSpanier}
\bibinfo{author}{Dall'Aglio, M.}, \bibinfo{year}{2001}.
\newblock \bibinfo{title}{{The Dubins-Spanier optimization problem in fair
  division theory}}.
\newblock \bibinfo{journal}{Journal of Computational and Applied Mathematics}
  \bibinfo{volume}{130}, \bibinfo{pages}{17--40}.
\bibitem[{Dall'Aglio and Maccheroni(2009)}]{DallAglio2009Disputed}
\bibinfo{author}{Dall'Aglio, M.}, \bibinfo{author}{Maccheroni, F.},
  \bibinfo{year}{2009}.
\newblock \bibinfo{title}{{Disputed lands}}.
\newblock \bibinfo{journal}{Games and Economic Behavior} \bibinfo{volume}{66},
  \bibinfo{pages}{57--77}.
\bibitem[{Dubins and Spanier(1961)}]{Dubins_1961}
\bibinfo{author}{Dubins, L.E.}, \bibinfo{author}{Spanier, E.H.},
  \bibinfo{year}{1961}.
\newblock \bibinfo{title}{{How to Cut A Cake Fairly}}.
\newblock \bibinfo{journal}{The American Mathematical Monthly}
  \bibinfo{volume}{68}.
\bibitem[{Even and Paz(1984)}]{Even1984}
\bibinfo{author}{Even, S.}, \bibinfo{author}{Paz, A.}, \bibinfo{year}{1984}.
\newblock \bibinfo{title}{{A Note on Cake Cutting}}.
\newblock \bibinfo{journal}{Discrete Applied Mathematics} \bibinfo{volume}{7},
  \bibinfo{pages}{285--296}.
\bibitem[{Fink(1964)}]{Fink1964}
\bibinfo{author}{Fink, A.}, \bibinfo{year}{1964}.
\newblock \bibinfo{title}{{A Note on the Fair Division Problem}}.
\newblock \bibinfo{journal}{Mathematics Magazine} \bibinfo{volume}{37},
  \bibinfo{pages}{341--342}.
\bibitem[{Ghodsi et~al.(2011)Ghodsi, Zaharia, Hindman, Konwinski, Shenker and
  Stoica}]{Ghodsi2011Dominant}
\bibinfo{author}{Ghodsi, A.}, \bibinfo{author}{Zaharia, M.},
  \bibinfo{author}{Hindman, B.}, \bibinfo{author}{Konwinski, A.},
  \bibinfo{author}{Shenker, S.}, \bibinfo{author}{Stoica, I.},
  \bibinfo{year}{2011}.
\newblock \bibinfo{title}{{Dominant Resource Fairness: Fair Allocation of
  Multiple Resource Types}}, in: \bibinfo{booktitle}{Proceedings of the 8th
  USENIX Conference on Networked Systems Design and Implementation},
  \bibinfo{publisher}{USENIX Association}, \bibinfo{address}{Berkeley, CA,
  USA}. pp. \bibinfo{pages}{323--336}.
\bibitem[{Herreiner and Puppe(2009)}]{Herreiner2009}
\bibinfo{author}{Herreiner, D.}, \bibinfo{author}{Puppe, C.},
  \bibinfo{year}{2009}.
\newblock \bibinfo{title}{{Envy Freeness in Experimental Fair Division
  Problems}}.
\newblock \bibinfo{journal}{Theory and Decision} \bibinfo{volume}{67},
  \bibinfo{pages}{65--100}.
\bibitem[{H\"{u}sseinov(2011)}]{Husseinov2011Theory}
\bibinfo{author}{H\"{u}sseinov, F.}, \bibinfo{year}{2011}.
\newblock \bibinfo{title}{{A theory of a heterogeneous divisible commodity
  exchange economy}}.
\newblock \bibinfo{journal}{Journal of Mathematical Economics}
  \bibinfo{volume}{47}, \bibinfo{pages}{54--59}.
\bibitem[{Jones(2002)}]{Jones2002}
\bibinfo{author}{Jones, M.A.}, \bibinfo{year}{2002}.
\newblock \bibinfo{title}{{Equitable, Envy-Free, and Efficient Cake Cutting for
  Two People and Its Application to Divisible Goods}}.
\newblock \bibinfo{journal}{Mathematics Magazine} \bibinfo{volume}{75},
  \bibinfo{pages}{275--283}.
\bibitem[{Legut et~al.(1994)Legut, Potters and Tijs}]{Legut1994Economies}
\bibinfo{author}{Legut, J.}, \bibinfo{author}{Potters, J.A.M.},
  \bibinfo{author}{Tijs, S.H.}, \bibinfo{year}{1994}.
\newblock \bibinfo{title}{{Economies with Land - A Game Theoretical Approach}}.
\newblock \bibinfo{journal}{Games and Economic Behavior} \bibinfo{volume}{6},
  \bibinfo{pages}{416--430}.
\bibitem[{Mawet et~al.(2010)Mawet, Pereira and Petit}]{mawet2010equitable}
\bibinfo{author}{Mawet, S.}, \bibinfo{author}{Pereira, O.},
  \bibinfo{author}{Petit, C.}, \bibinfo{year}{2010}.
\newblock \bibinfo{title}{Equitable cake cutting without mediator}, in:
  \bibinfo{booktitle}{5th Benelux Workshop on Information and System Security}.
\bibitem[{Moulin(2004)}]{Moulin2004Fair}
\bibinfo{author}{Moulin, H.}, \bibinfo{year}{2004}.
\newblock \bibinfo{title}{{Fair Division and Collective Welfare}}.
\newblock \bibinfo{publisher}{The MIT Press}.
\bibitem[{Nicol\`{o} et~al.(2012)Nicol\`{o}, Perea and
  Roberti}]{Nicolo2012Equal}
\bibinfo{author}{Nicol\`{o}, A.}, \bibinfo{author}{Perea},
  \bibinfo{author}{Roberti, P.}, \bibinfo{year}{2012}.
\newblock \bibinfo{title}{{Equal opportunity equivalence in land division}}.
\newblock \bibinfo{journal}{SERIEs - Journal of the Spanish Economic
  Association} \bibinfo{volume}{3}, \bibinfo{pages}{133--142}.
\bibitem[{Peleg and Sudh\"olter(2007)}]{Peleg2007}
\bibinfo{author}{Peleg, B.}, \bibinfo{author}{Sudh\"olter, P.},
  \bibinfo{year}{2007}.
\newblock \bibinfo{title}{Introduction to the Theory of Cooperative Games}.
\newblock \bibinfo{publisher}{Springer}, \bibinfo{address}{Heidelberg}.
\bibitem[{Procaccia and Wang(2017)}]{Procaccia2017}
\bibinfo{author}{Procaccia, A.}, \bibinfo{author}{Wang, J.},
  \bibinfo{year}{2017}.
\newblock \bibinfo{title}{A Lower Bound for Equitable Cake Cutting}.
\newblock \bibinfo{type}{Working Paper} \bibinfo{number}{W2/2017}. Carnegie
  Mellon University.
\bibitem[{Simmons and Su(2003)}]{Simmons2003}
\bibinfo{author}{Simmons, F.W.}, \bibinfo{author}{Su, F.E.},
  \bibinfo{year}{2003}.
\newblock \bibinfo{title}{{Consensus-halving via theorems of Borsuk-Ulam and
  Tucker}}.
\newblock \bibinfo{journal}{Mathematical Social Sciences} \bibinfo{volume}{45},
  \bibinfo{pages}{15--25}.
\bibitem[{Steinhaus(1948)}]{Steinhaus1948}
\bibinfo{author}{Steinhaus, H.}, \bibinfo{year}{1948}.
\newblock \bibinfo{title}{{The problem of fair division}}.
\newblock \bibinfo{journal}{Econometrica} \bibinfo{volume}{16},
  \bibinfo{pages}{101--104}.
\bibitem[{Sziklai and Segal-Halevi(2015)}]{ourArxivPaperAdditive}
\bibinfo{author}{Sziklai, B.}, \bibinfo{author}{Segal-Halevi, E.},
  \bibinfo{year}{2015}.
\newblock \bibinfo{title}{{Resource-monotonicity and Population-monotonicity in
  Cake-cutting}}.
\newblock \eprint{arXiv preprint 1510.05229}.
\bibitem[{Tasn\'adi(2003)}]{Tasnadi2003}
\bibinfo{author}{Tasn\'adi, A.}, \bibinfo{year}{2003}.
\newblock \bibinfo{title}{{A new proportional procedure for the n-person
  cake-cutting problem}}.
\newblock \bibinfo{journal}{Economics Bulletin} \bibinfo{volume}{4},
  \bibinfo{pages}{1--3}.
\bibitem[{Thomson(1997)}]{Thomson1997Replacement}
\bibinfo{author}{Thomson, W.}, \bibinfo{year}{1997}.
\newblock \bibinfo{title}{{The Replacement Principle in Economies with
  Single-Peaked Preferences}}.
\newblock \bibinfo{journal}{Journal of Economic Theory} \bibinfo{volume}{76},
  \bibinfo{pages}{145--168}.
\bibitem[{Thomson(2011)}]{Thomson_2011}
\bibinfo{author}{Thomson, W.}, \bibinfo{year}{2011}.
\newblock \bibinfo{title}{{Fair Allocation Rules}}, in:
  \bibinfo{booktitle}{Handbook of Social Choice and Welfare}.
  \bibinfo{publisher}{Elsevier {BV}}, pp. \bibinfo{pages}{393--506}.
\bibitem[{Thomson(2012)}]{Thomson_2012}
\bibinfo{author}{Thomson, W.}, \bibinfo{year}{2012}.
\newblock \bibinfo{title}{{On The Axiomatics Of Resource Allocation:
  Interpreting The Consistency Principle}}.
\newblock \bibinfo{journal}{Economics and Philosophy} \bibinfo{volume}{28},
  \bibinfo{pages}{385--421}.
\bibitem[{Varian(1974)}]{Varian1974Equity}
\bibinfo{author}{Varian, H.R.}, \bibinfo{year}{1974}.
\newblock \bibinfo{title}{{Equity, Envy, and Efficiency}}.
\newblock \bibinfo{journal}{Journal of Economic Theory} .
\bibitem[{Walsh(2011)}]{Walsh2011Online}
\bibinfo{author}{Walsh, T.}, \bibinfo{year}{2011}.
\newblock \bibinfo{title}{{Online Cake Cutting}}.
\newblock \bibinfo{journal}{Algorithmic Decision Theory}
  \bibinfo{volume}{6992}, \bibinfo{pages}{292--305}.
\bibitem[{Young(1987)}]{Young1987}
\bibinfo{author}{Young, H.}, \bibinfo{year}{1987}.
\newblock \bibinfo{title}{On dividing an amount according to individual claims
  or liabilities}.
\newblock \bibinfo{journal}{Mathematics of Operations Research}
  \bibinfo{volume}{12}, \bibinfo{pages}{65--72}.

\end{thebibliography}
}
\end{document}